\documentclass[preprint,journal]{vgtc}

\title{Scalable Hypergraph Visualization}

\author{
Peter~Oliver,
Eugene~Zhang, \textit{Senior Member,~IEEE,}
and Yue~Zhang, \textit{Member,~IEEE}}

\authorfooter{
\item
Peter Oliver is with the School of Electrical Engineering and Computer Science, Oregon State University. E-mail: oliverpe@oregonstate.edu.
\item
Eugene Zhang is a Professor with the School of Electrical Engineering and Computer Science, Oregon State University. E-mail: zhange@eecs.oregonstate.edu.
\item
Yue Zhang is an Associate Professor with the School of Electrical Engineering and Computer Science, Oregon State University. E-mail: zhangyue@oregonstate.edu.
}

\shortauthortitle{Oliver \MakeLowercase{\textit{et al.}}: Scalable Hypergraph Visualization}

\abstract{Hypergraph visualization has many applications in network data analysis. Recently, a polygon-based representation for hypergraphs has been proposed with demonstrated benefits. However, the polygon-based layout often suffers from excessive self-intersections when the input dataset is relatively large. In this paper, we propose a framework in which the hypergraph is iteratively simplified through a set of atomic operations. Then, the layout of the simplest hypergraph is optimized and used as the foundation for a reverse process that brings the simplest hypergraph back to the original one, but with an improved layout. At the core of our approach is the set of atomic simplification operations and an operation priority measure to guide the simplification process. In addition, we introduce necessary definitions and conditions for hypergraph planarity within the polygon representation. We extend our approach to handle simultaneous simplification and layout optimization for both the hypergraph and its dual. We demonstrate the utility of our approach with datasets from a number of real-world applications.}

\keywords{Hypergraph visualization, scalable visualization, polygon layout, hypergraph embedding, primal-dual visualization}

\teaser{
\subfloat[][Qu et al.~\cite{Qu:22}]{\raisebox{0.375in}{\includegraphics[width=1.25in]{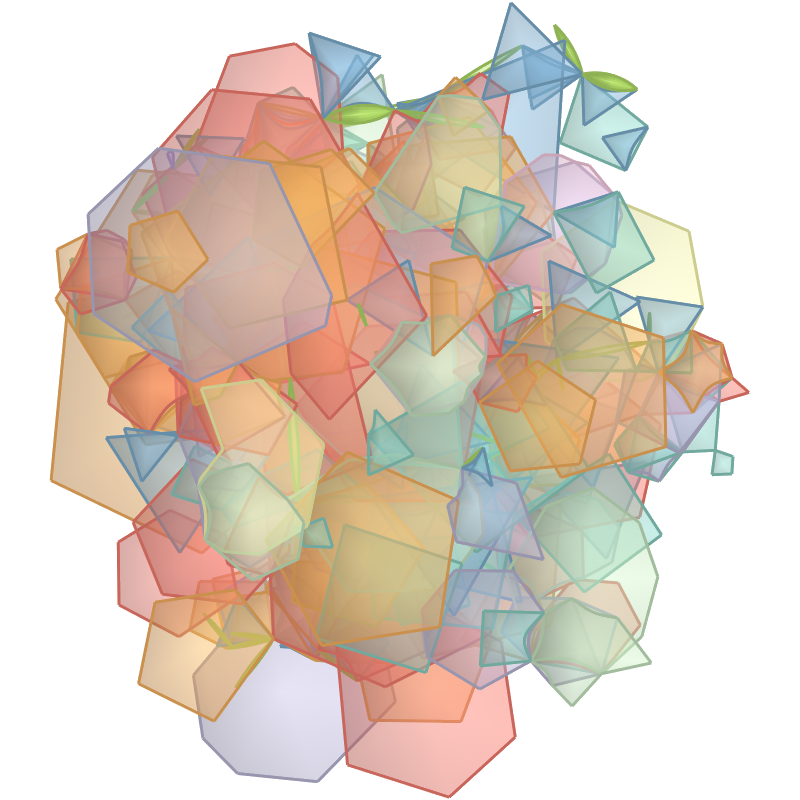}}}
\hspace{0.25in}
\raisebox{0.25in}{\includegraphics[height=1.5in]{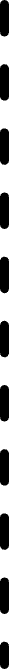}}
\hspace{0.25in}
\subfloat[][Our method: intermediate scale]{\raisebox{0.0625in}{\includegraphics[width=2.25in]{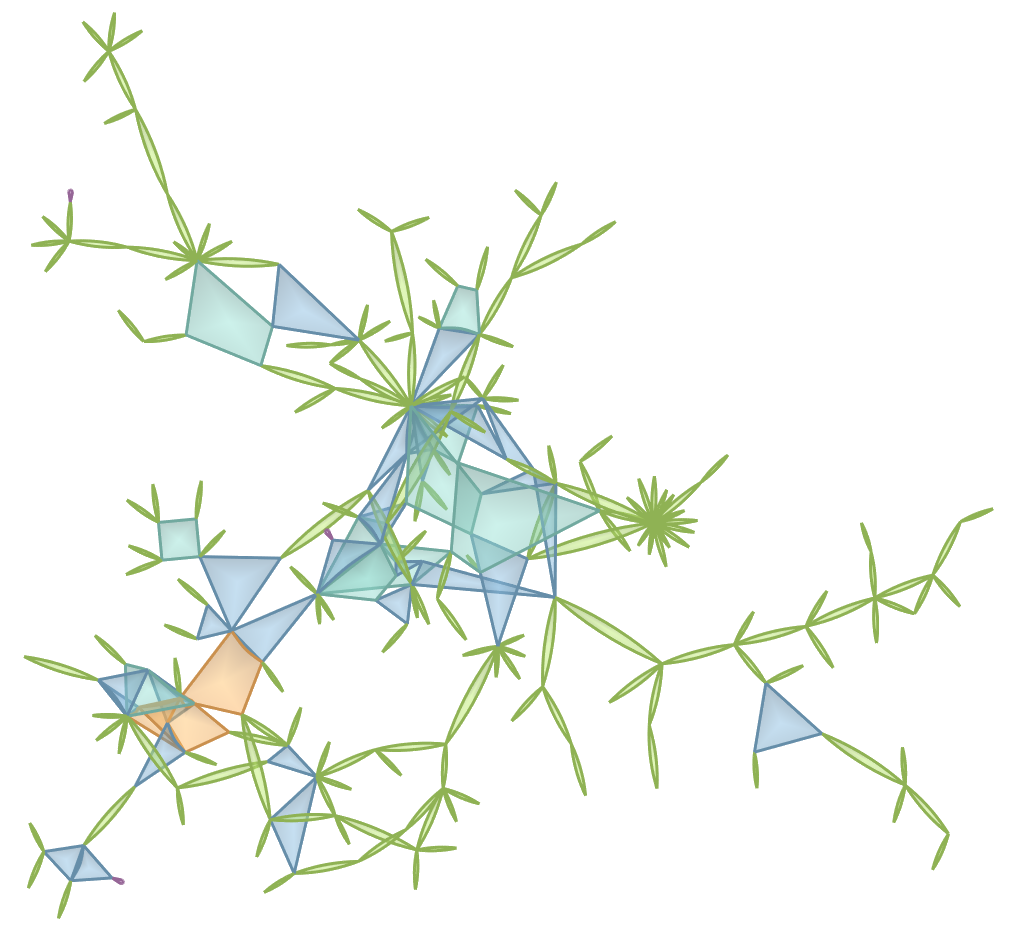}}}
\hspace{0.0625in}
\subfloat[][Our method: final optimized layout]{\includegraphics[width=2.25in]{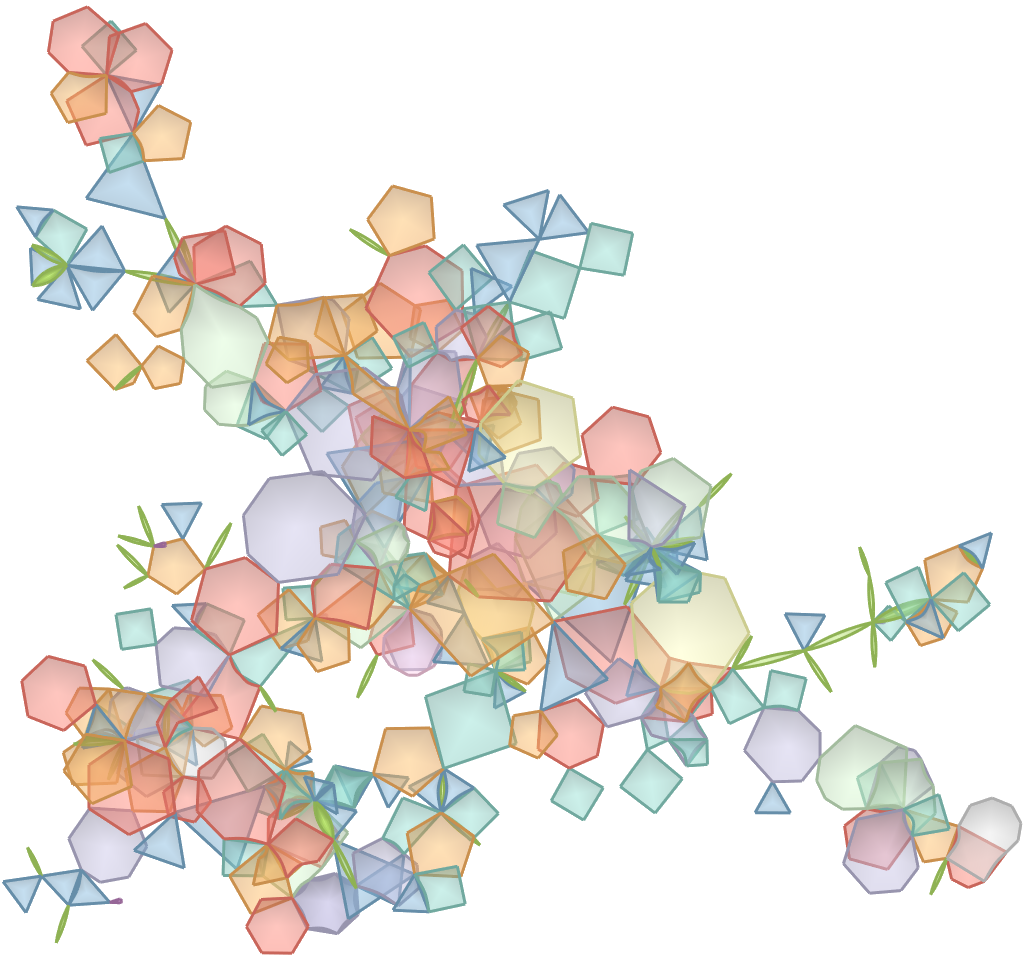}}
\caption{A paper-author hypergraph network with $786$ vertices and $318$ hyperedges. We first show the result of Qu et al.~\cite{Qu:22} in (a). Using our framework, the same hypergraph is simplified before the layout optimization begins (b). Then the simplification is iteratively reversed and the layout is refined until an optimized layout for the original hypergraph is recovered in (c).}
\label{fig:teaser}
}

\ifpdf
\pdfoutput=1\relax
\usepackage{graphicx}
\DeclareGraphicsExtensions{.pdf,.png,.jpg,.jpeg}
\else
\usepackage{graphicx}
\DeclareGraphicsExtensions{.eps}
\fi

\usepackage{microtype}
\PassOptionsToPackage{warn}{textcomp}
\usepackage{textcomp}
\usepackage{mathptmx}
\usepackage{times}

\usepackage{cite}
\usepackage{tabu}
\usepackage{booktabs}
\usepackage{mathptmx}
\usepackage{wrapfig}
\usepackage{todonotes}
\usepackage{graphicx}
\usepackage{times}
\usepackage{amsmath,amscd,amssymb, amsthm}
\usepackage{times}
\usepackage[abs]{overpic}
\usepackage{cases}
\usepackage{float}
\usepackage[font={scriptsize,sf}]{caption,subfig}
\usepackage{wrapfig}
\usepackage{placeins}
\usepackage{multirow}
\usepackage{paralist}
\usepackage{enumitem}
\usepackage{helvet}
\usepackage{colortbl}
\usepackage{threeparttable}

\long\def\symbolfootnote[#1]#2{\begingroup%
	\def\thefootnote{\fnsymbol{footnote}}\footnote[#1]{#2}\endgroup}

\newtheorem{theorem}{\sffamily Theorem}
\newtheorem{lemma}[theorem]{\sffamily Lemma}

\newtheorem{definition}[theorem]{\sffamily Definition}

\newtheorem{procedure}[theorem]{\sffamily Procedure}

\vgtcinsertpkg

\RequirePackage[nameinlink,capitalise]{cleveref}
\crefname{table}{Tab.}{Tabs.}
\Crefname{table}{Table}{Tables}
\crefname{section}{Sec.}{Secs.}
\Crefname{section}{Section}{Sections}

\begin{document}

\firstsection{Introduction}\label{sec:introduction}

\maketitle

Hypergraphs are a generalization of graph data structures consisting of a set of vertices and a family of hyperedges. A hyperedge joins any number of $n\geq 1$ vertices and provides a natural way to represent \emph{polyadic} (multi-sided) relationships~\cite{Qu:2017}. Hypergraphs can be thought of as networks of polyadic relationships and have many applications in social sciences, biology, computer science, and engineering where such relationships are prevalent~\cite{Alsallakh:16}.

Hypergraph visualization has seen many advances in recent decades~\cite{Alsallakh:16} with a focus on finding a proper visual metaphor for representing hyperedges (polyadic relationships) to facilitate a number of common analysis tasks. Qu et al.~\cite{Qu:2017} introduce a visual metaphor in which each hyperedge takes the form of a 2D polygon in the plane, whose vertices encode the members of the underlying polyadic relationship. This representation allows the cardinality of a hyperedge (number of vertices) to be easily understood. For example, a paper-author hypergraph dataset, in which each vertex represents an author and each hyperedge a research paper, can be visualized with the polygon metaphor to easily communicate the number of co-authors for each publication (\Cref{fig:teaser}). Qu et al.~\cite{Qu:22} also develop an optimization framework that can automatically generate a high-quality polygon layout for a hypergraph with tens of hyperedges based on a set of visualization design principles that they identify. Recognizing the duality between the vertices and the hyperedges in a hypergraph, they augment their optimization framework to simultaneously generate high-quality layouts for the input hypergraph and its \textit{dual hypergraph} in which the roles of the vertices and hyperedges are reversed.

However, scalability presents a major challenge to their approach for large hypergraph datasets, which can have hundreds or thousands of vertices and hyperedges. With a relatively large dataset, their optimization process can be trapped at local minima, producing suboptimal layouts with excessive overlaps between polygons (\Cref{fig:teaser} (a)). To address this challenge, we introduce a new polygon-based layout optimization framework in which a complex hypergraph is automatically simplified by iteratively applying a set of atomic simplification operations that we have identified. The simplification process terminates when one or more user-specified criteria are met. Next, we generate a polygon layout for the simplified hypergraph using a version of the optimization method of Qu et al.~\cite{Qu:22} which we modify to make efficient use of their optimization energy terms. From there, our framework iteratively inverts the simplification operations until the input hypergraph is recovered. Each time a simplification is inverted, the layout is locally optimized over a neighborhood immediately surrounding the location of the operation. Our framework (\Cref{fig:teaser} (c)) leads to an improved final layout of the original hypergraph compared to the framework of Qu et al.~\cite{Qu:22} (\Cref{fig:teaser} (a)). We extend our framework to handle simultaneous layout optimization of the hypergraph and its dual hypergraph, taking advantage of the fact that our atomic operations naturally simplify both hypergraphs. This is demonstrated in one of our case study examples in \Cref{sec:casestudies} (\Cref{fig:trade}). During the simplification step, the order of atomic operations is determined by a priority measure which we design to reduce the amount of overlap among polygons and preserve local structures such as high-degree and centrally located elements in the visualization.

\setlength{\belowcaptionskip}{-2pt}
\begin{figure}[tbp]
  \centering
  \subfloat[][]{\includegraphics[width=1.0in]{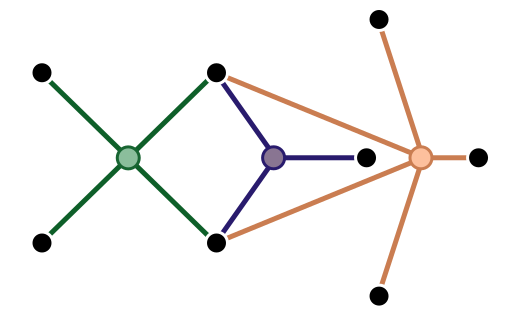}} \hspace{0.125in}
  \subfloat[][]{\includegraphics[width=1.0in]{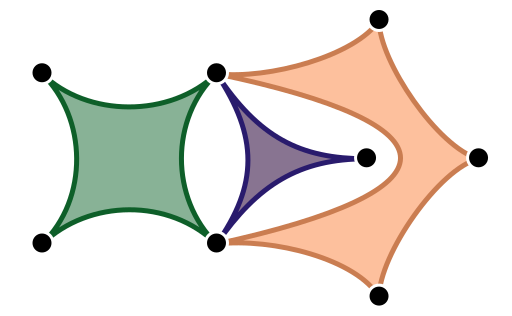}} \hspace{0.125in}
  \subfloat[][]{\includegraphics[width=1.0in]{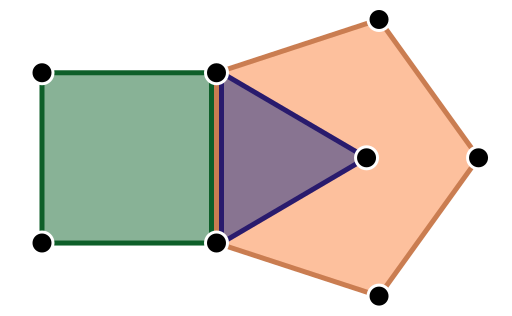}}
  \caption{A hypergraph with a planar K{\"o}nig representation (a) is plane embeddable when the regions can be represented by arbitrary shapes as in Zykov's representation (b). However, with the additional requirement of polygon convexity (c), the hypergraph has unavoidable overlaps using the polygon representation~\cite{Qu:22}.
  }\label{fig:planarity}
\end{figure}

As the polygon-based metaphor requires all polygons to be convex, a hypergraph may not be plane-embeddable even when it is planar (\Cref{fig:planarity}). That is, the requirement of convexity greatly reduces the set of hypergraphs that can be mapped to the plane without self-overlap in the regions representing the hyperedges. We investigate this issue and introduce a new notion of hypergraph planarity with convex polygons. Being able to detect subsets of non-planar hyperedges allows us to save time on attempting to remove overlaps between such hyperedges.

Our framework exceeds the performance of~\cite{Qu:22} for hypergraph datasets with more than $1000$ elements in terms of reducing polygon overlaps which is crucial to visual clarity. To enable our priority-guided simplification, we also introduce a new vertex and hyperedge based statistic called \emph{adjacency factor} which correlates to non-planar sub-hypergraphs.

We demonstrate the utility of our framework with two applications: (1) a paper-author collaboration network and (2) a network of international trade agreements. To evaluate the effectiveness of our layout framework, we conduct a user survey where participants have completed analysis tasks using our final optimized layouts as well as a few scales of simplification. We utilize eye-tracking technology to study participants' exploration of our visualizations while they answered task-driven questions. The preliminary results suggest that our new layout method allowed the survey participants to perform the tasks with relatively high accuracy.

We make the following contributions to hypergraph visualization:
\begin{enumerate}
  \item A novel multi-scale optimization framework for generating high-quality polygon-based visualizations of hypergraphs with thousands of vertices and hyperedges.
  \item A novel priority-guided hypergraph simplification method which is the first to operate on both vertices and hyperedges.
  \item A set of atomic simplification operations which can simplify a hypergraph and its dual hypergraph simultaneously.
  \item A new definition for hypergraph planarity within the polygon visualization metaphor.
\end{enumerate}

\section{Related Work}
\label{sec:prev_work}

In this section, we review past research in graph and hypergraph visualization that is most relevant to our work.

\subsection{Hypergraph Visualization}

Hypergraph visualization has been well explored during recent decades~\cite{Alsallakh:16}. Much of this research has focused on identifying the visual representation of hyperedges, such as matrices~\cite{Kim:2007,Sadana:2014,Lex:2014,valdivia2019analyzing}, bipartite graphs~\cite{Stasko:07,Dork:12,Alsallakh:2013}, and metro lines~\cite{Wu:2020,Jacobsen:2021,Frank:2021}. Region-based visual metaphors derived from Euler and Venn diagrams~\cite{Rogers:08,simonetto2009fully,Stapleton:12,Micallef:14}, represent sets (hyperedges) as closed regions whose overlaps indicate the intersections of their corresponding sets. The vertices in the hypergraph are often not explicitly shown, such as~\cite{Rogers:08}. More recent approaches explicitly represent set elements (vertices) by drawing them as points inside the corresponding regions~\cite{Santamara:2010,Riche:2011,Alsallakh:2013,Arafat:17,simonetto2015simple}. As pointed out in~\cite{Qu:22}, placing the vertices inside the regions can make it difficult to identify the cardinality of the hyperedges. Instead, Zykov~\cite{zykov1974hypergraphs} restricts vertex placement to the boundaries of the regions. Qu et al.~\cite{Qu:2017} represent each hyperedge as a polygon so the vertices of the hyperedge are also the vertices of the polygon. Unlike Zykoy's approach where the region can take arbitrary shapes, Qu et al.~\cite{Qu:2017} require the polygons to be as close to regular as possible and thus convex. With this representation, identifying the cardinality of a hyperedge is the same as recognizing the cardinality of the corresponding polygon. Qu et al.~\cite{Qu:22} identify a number of design principles for polygon-based hypergraph drawings and develop an automatic layout optimization system based on these principles.

However, the objective functions used in~\cite{Qu:22} are not convex, thus leading to local minimums that make the final hypergraph layouts suboptimal, especially for large datasets. In addition, some of the hypergraphs cannot be embedded in the plane without overlaps using the polygon representation, even when they are plane embeddable if the hyperedges are represented by arbitrary (possibly non-convex) shapes. In this paper, we introduce a new multi-scale optimization framework that can lead to improved hypergraph layouts compared to those from~\cite{Qu:22}. Furthermore, we introduce the notion of polygon planarity, which can save on computation attempting to remove overlaps among hyperedges that are inevitable due to the polygon convexity requirement.

\subsection{Graph and Hypergraph Simplification}

Techniques for reducing complexity in graphs have been well studied and provide numerous advantages for improving graph-based algorithm efficiency and graph visualization. Depending on the application, it may be more valuable to reduce the number of graph vertices (coarsening), or the number of edges (sparsification)~\cite{bravo:2019:unifying}.

Graph sparsification algorithms have been studied extensively and two main categories of graph sparsifiers have arisen: \textit{cut sparsifiers} and \textit{spectral sparsifiers}. We review only the most relevant works here. Bencz\'{u}r and Karger~\cite{benczur:1996} introduce cut sparsifiers which approximate every cut in a weighted graph to an arbitrarily small multiplicative error. Spielman and Teng~\cite{spielman:2011:spectral} introduce the stronger notion of spectral sparsifiers which approximate the Laplacian quadratic form of the graph to an arbitrarily small multiplicative error.

Graph coarsening has been primarily used to construct multi-level graph frameworks for graph partitioning problems. Such frameworks transform an input graph $G_0$ into a sequence of smaller graphs $G_1,G_2,...,G_n$ such that each level in the sequence contains fewer vertices than the previous graph. This is usually accomplished through a graph coarsening scheme in which a set of vertices in $G_i$ is merged into a single \textit{multi-node} in the next coarser level $G_{i+1}$. Identifying appropriate vertex sets for merging has followed two main approaches: vertex pair matching~\cite{bui:1993:heuristic,hendrickson:1995:multi,karypis:1998:fast} and vertex grouping based on some graph-based statistics such as high connectivity or affinity~\cite{cheng:1991:improved,garbers:1990:finding,hagen:1991:fast,hagen:1992:new,ron:2011:relaxation,imre:2020:spectrum}.

Several frameworks combine graph sparsification with multi-level coarsening to reduce the number of vertices and edges in a graph. Imre et al.~\cite{imre:2020:spectrum} perform sparsification and coarsening in two separate phases of their algorithm while Bravo-Hermsdorff and Gunderson~\cite{bravo:2019:unifying} present a unified framework that incorporates vertex deletion, vertex contraction, edge deletion, and edge contraction.

Hypergraph sparsification is less studied but has been gaining traction in recent years. Cut sparsifier algorithms have been extended to hypergraphs with near-linear time complexities~\cite{kogan:2015:sketching,chekuri:2018:minimum}, and most recently with a sub-linear time complexity~\cite{chen:2021:hypergraph}. The notion of the Laplacian for undirected hypergraphs is introduced by Louis~\cite{louis:2015:hypergraph}, which has recently been used to design algorithms for producing linear hypergraph spectral sparcifiers~\cite{kapralov:2022:spectral,soma:2019:spectral}.

Multi-level coarsening has also been extended to hypergraph partitioning~\cite{alpert:1996:hybrid,hauck:1997:evaluation,cong:1993:parallel,karypis:1999:multilevel}. Alpert et al.~\cite{alpert:1996:hybrid} first convert the hypergraph to a graph by replacing each hyperedge with a graph clique and applying existing graph coarsening schemes. Karypis et al.~\cite{karypis:1999:multilevel} develop a coarsening scheme that acts directly on the hypergraph in which vertices belonging to selected hyperedges are merged together. More recent applications of hypergraph coarsening have presented distributed hypergraph partitioners using parallel versions of the multi-level technique~\cite{devine:2006:parallel,trifunovic:2008:parallel,kabiljo:2017:social}.

To our knowledge, we are the first to present a unified hypergraph simplification framework that operates on both hypergraph vertices and hyperedges. In addition, we introduce the notion of polygon planarity for hypergraphs and develop a simplification priority function that aims to preserve structures in the hypergraphs as well as reduce unnecessary polygon overlaps.

\section{Background and Notations}
\label{sec:background}

\begin{figure}[t]
  \centering
  \vspace{0.125in}
  \includegraphics[width=1.625in]{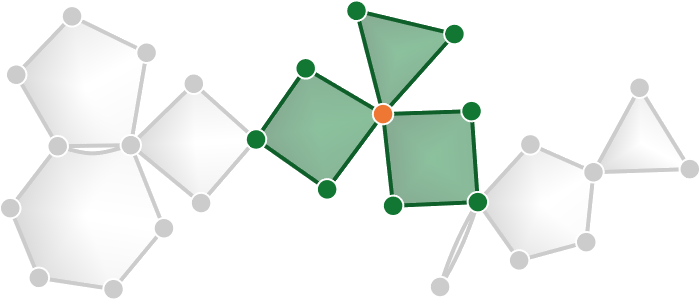}
  \hspace{0.125in}
  \includegraphics[width=1.625in]{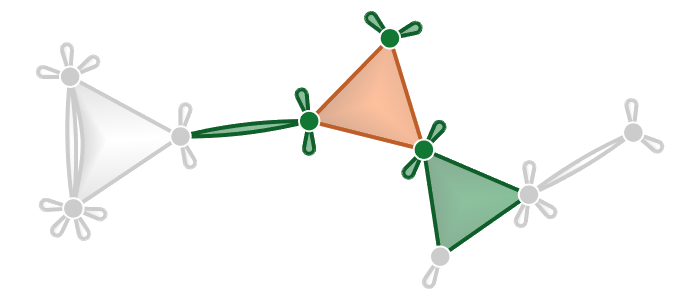}
  \vspace{0.0625in}
  \caption{A primal hypergraph (left) and its dual (right). The neighborhood of a vertex in the primal hypergraph (left: the orange dot) corresponds to the neighborhood of its dual hyperedge in the dual hypergraph (right: the orange triangle).}\label{fig:neighborhood}
\end{figure}

Following the terminology of Berge~\cite{berge1973graphs} and Bretto~\cite{bretto2013hypergraph}, a hypergraph $H=\langle V,E \rangle$ on a finite set of $n$ vertices $V$ is defined by a family of $m$ hyperedges $E$. A hyperedge $e \in E$ contains a non-empty subset of vertices $V_e \subseteq V$ which we say are \emph{incident} to $e$ and \emph{adjacent} to each other. Similarly, a vertex $v \in  V$ is contained by a subset of hyperedges $E_v \subseteq E$ which we say are \emph{incident} to $v$ and \emph{adjacent} to each other. Let $E_e$ denote the set of hyperedges adjacent to $e$ and $V_v$ the set of vertices adjacent to $v$. $H$ is \emph{complete} if all vertices in $V$ are adjacent to each other and \emph{linear} if $|V_e \cap V_f| \leq 1$ for all $e \ne f \in E$. $H$ is \textit{connected} if there exists an alternating sequence of vertices and hyperedges connecting each pair of distinct vertices in $H$. The hypergraphs we consider in this paper are assumed to be connected unless otherwise specified. Consistent with \cite{Qu:22}, we define the \emph{degree} of a vertex $v$ as $\text{deg}(v) = |E_v|$ and the \emph{cardinality} of a hyperedge $e$ as $\text{card}(e) = |V_e|$. Notice that the traditional notion of a graph is simply a hypergraph where every hyperedge has cardinality two.

The \emph{dual hypergraph} $H'=\langle V',E'\rangle$ of $H$ is obtained by swapping the roles of vertices and hyperedges in $H$. For convenience, we call the original hypergraph $H$ the \emph{primal hypergraph}. More precisely, each element $v \in V$ corresponds to a unique element $v' \in E'$ and each element $e \in E$ corresponds to a unique element $e' \in V'$. Furthermore, the incidence and adjacency relationships of corresponding elements in the primal and dual hypergraphs are identical. This means that the degree of a vertex $v \in V$ is the same as the cardinality of the corresponding hyperedge $v' \in E'$ and vice versa. Thus, the dual of a linear hypergraph is also linear~\cite{berge1973graphs}.

For a set of vertices $A \subseteq V$ and a set of hyperedges $J \subseteq E$, Berge~\cite{berge1973graphs} defines the \emph{sub-hypergraph induced by \textit{A}} and the \emph{partial hypergraph generated by \textit{J}} as respectively,
\[H_A = \left\langle A, \{e \cap A \> | \> e \in E, \> e \cap A \ne \varnothing\} \right\rangle \hspace{0.125in} \text{and} \hspace{0.125in} H_J = \langle V_J \subseteq V, \> J \rangle.\]
For a vertex $v \in V$, we call the partial hypergraph generated by the hyperedges incident to $v$, $H_{E_v}$, with vertex set $V_v \cup \{v\}$ the \emph{neighborhood~of~\textit{v}}. We define the neighborhood of a hyperedge $e \in E$ to be the sub-hypergraph induced by the vertices incident to $e$, $H_{V_e}$. In other words, the neighborhood of a vertex or hyperedge simply consists of all its incident and adjacent elements (\Cref{fig:neighborhood}).

For a hypergraph $H=\langle V,E \rangle$, the K{\"o}nig representation ${K(H)=(X \cup Y,D)}$ is a bipartite graph with vertices $x \in X$ for every hypergraph vertex $v \in V$ and vertices $y \in Y$ for every hyperedge $e \in E$. Vertices $x \in X$ and $y \in Y$ form an edge $(x,y) \in D$ only if the hypergraph vertex corresponding to $x$ and the hyperedge corresponding to $y$ are incident in $H$. The hypergraph $H$ is \emph{Zykov planar} if $K(H)$ is a planar graph.

\section{Hypergraph Simplification}

In this section, we describe the building blocks of our multi-scale hypergraph layout optimization framework: (1) the set of atomic hypergraph simplification operations and (2) a number of terms used in our simplification objectives.

\subsection{Atomic Operations}
\label{sec:operations}

Two guiding principles for generating high-quality polygon layouts of hypergraphs are to maximize the regularity of each polygon and minimize overlap between the polygons~\cite{Qu:22}. As such, non-planar hypergraphs are difficult to handle through optimization since their polygon layouts contain necessary overlap leading to local minima in the optimization space. Non-planar sub-hypergraphs tend to appear in clusters of hyperedges that share common vertices and are also caused by structures analogous to $K_5$ and $K_{3,3}$ from graph theory. Certain structures in planar hypergraphs are also prone to overlaps, such as high-degree vertices that do not have enough angular space around them to draw each of their incident hyperedges as non-overlapping regular polygons. In these situations, a conflict between overlap minimization and regularity maximization makes optimization more difficult and requires that one or both of these objectives be compromised in the final results (\Cref{fig:entangled}). A core idea behind our multi-scale layout optimization approach is to use simplification to reduce the challenging configurations in both planar and non-planar portions of the input hypergraph and avoid optimization that terminates prematurely. To achieve this, we present a set of hypergraph simplification operations specifically designed to eliminate these challenging configurations.

\begin{figure}[t]
  \centering
  \subfloat[][]{\includegraphics[height=.75in]{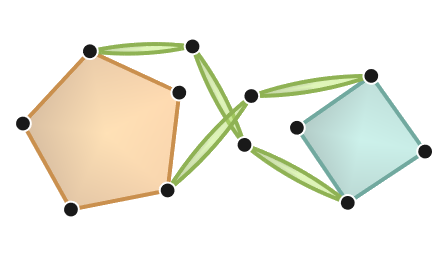}}
  \hspace{0.125in}
  \subfloat[][]{\includegraphics[height=.75in]{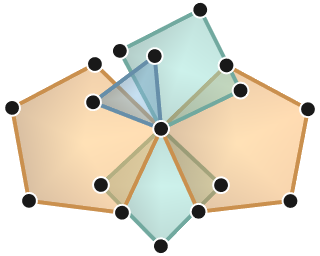}}
  \hspace{0.125in}
  \subfloat[][]{\includegraphics[height=.75in]{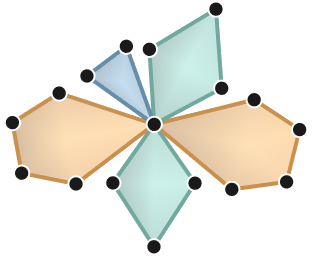}}
  \caption{Examples of avoidable overlaps in layouts of hypergraphs that have a convex polygon representation. In (a), un-twisting the layout would require making one of the polygons temporarily irregular, so the optimization halts before overlap can be resolved. In (b) and (c), there is not enough angular space around the central vertex for all the hyperedges to be drawn as regular polygons. In (b), the overlap minimization objective is compromised. In (c), the regularity maximization objective is compromised.}
  \label{fig:entangled}
\end{figure}

We identify four atomic operations for simplifying a hypergraph $H$:

\begin{enumerate}[nosep]
  \vspace{2pt}
  \itemsep1pt
  \item {\em Vertex removal:} a vertex is removed from $H$.
  \item {\em Hyperedge removal:} a hyperedge is removed from $H$.
  \item {\em Vertex merger:} a pair of adjacent vertices are combined into a single vertex whose set of incident hyperedges is the union of the two inputs.
  \item {\em Hyperedge merger:} a pair of adjacent hyperedges are merged into a single hyperedge whose set of incident vertices is the union of the two inputs.
\end{enumerate}

The vertex removal and hyperedge removal operations form a primal-dual pair in the sense that applying one to the primal hypergraph is equivalent to applying the other to the dual hypergraph (\Cref{fig:operations} (a,b)). The vertex merger and hyperedge merger operations similarly form a primal-dual pair (\Cref{fig:operations} (c,d)). We define the \textit{footprint} of an operation $O$ to be the union of the neighborhoods of its operand elements. For example, if $O$ merges two vertices $u,v \in V(H)$, the footprint of $O$ is given by $H_O = \langle V_u \cup V_v, E_u \cup E_v \rangle$.

Each atomic simplification operation has a corresponding inverse operation:
\begin{enumerate}[nosep]
  \vspace{2pt}
  \itemsep1pt
  \item {\em Vertex addition:} a removed vertex is added back into $H$.
  \item {\em Hyperedge addition:} a removed hyperedge is added back into $H$.
  \item {\em Vertex split:} a merged vertex is split into two vertices with one or more common hyperedges.
  \item {\em Hyperedge split:} a merged hyperedge is split into two hyperedges that contain one or more common vertices.
  \vspace{2pt}
\end{enumerate} These inverse operations are used to reverse simplification and similarly form primal-dual pairs.

A sequence of atomic simplification operations $\{O_1,O_2,\dots,O_n\}$ on a hypergraph $H$ defines a sequence of simplified scales $\{H_0,H_1,H_2,\dots,H_n\}$ where $H_0=H$ and $H_i = O_i(H_{i-1})$. Here $O_i(*)$ denotes applying operation $O_i$ to a hypergraph. In this multi-scale representation, we call $H_0$ the \textit{input} or \textit{original scale}, each $H_i = \langle V_i,E_i \rangle$ ($0 < i \leq n$) the \textit{i-th simplified scale}, and $H_n$ the \textit{coarsest simplified scale}. Given the nature of our atomic operations, each simplified scale is smaller than the previous scale, i.e., $|V(H_i)|+|E(H_i)| > |V(H_{i+1})|+ |E(H_{i+1})|$. By defining a prioritized sequence of atomic operations, we can construct a multi-scale representation where the size or number of non-planar sub-hypergraphs is reduced at each scale. In such a representation, it is generally easier to optimize the polygon layouts of successive simplified scales. This observation is central to our multi-scale layout optimization framework where we start by optimizing the coarsest simplified scale and handle the non-planar sub-hypergraphs on a localized basis while reversing simplification.

\subsection{Simplification Objectives}
\label{sec:statistics}

Numerous vertex and hyperedge based statistics could be used to guide the prioritization of atomic operations to achieve a variety of simplification objectives. We consider three statistics for this purpose: \textit{vertex degree} and \textit{hyperedge cardinality, betweenness centrality,} and \textit{adjacency factor}. Vertex degree and hyperedge cardinality are straightforward to compute based on the incidence relationships present in the hypergraph. Since the incidence relationships are identical between corresponding primal and dual hypergraph elements, the degree of a primal vertex is the same as the cardinality of its dual hyperedge and vice versa. By computing both vertex degree and hyperedge cardinality, we account for the primal and dual hypergraphs simultaneously. Vertex degree gives us an estimate of how much angular space is needed around the vertex for its incident polygons, and hyperedge cardinality gives us an estimate of how much area each hyperedge requires in an optimized polygon layout. As such, simplifying high-degree vertices and high-cardinality hyperedges can leave more space in the layout for neighboring elements and potentially reduce avoidable polygon overlaps. However, simplifying high-degree vertices and high-cardinality hyperedges may not be appropriate if they have an important semantic meaning in the underlying dataset.

\textit{Betweenness centrality} quantifies the proportion of shortest paths passing through a given vertex or hyperedge. Let $\sigma_{st} = \sigma_{ts}$ denote the number of shortest paths between $s,t \in V$, where $\sigma_{ss} = 1$ by convention. Let $\sigma_{st}(v)$ denote the number of shortest paths from $s$ to $t$ passing through $v \in V$. Then the betweenness centrality for $v$ is given by
\begin{equation*}
  C_B(v) = \sum_{s \neq v \neq t \in V}\frac{\sigma_{st}(v)}{\sigma_{st}}.
\end{equation*}
We compute the betweenness centrality of vertices and hyperedges simultaneously by applying the algorithm of Brandes~\cite{Brandes:2001} to the K{\"o}nig graph $K(H)$. To avoid an explicit summation in the betweenness centrality computation of each element, Brandes' algorithm leverages a recursive relationship for partial sums. Their algorithm is able to accumulate partial sums over a single depth-first search and return the betweenness centralities of each vertex. Avoiding simplification of elements with high betweenness centrality can be used to preserve the path structure in the input hypergraph and to preserve path-related features such as hypergraph cycles.

We define a new statistic, adjacency factor, to measure the volume of connections between a given vertex or hyperedge and its adjacent elements. It is an extension of adjacency as defined by Bretto~\cite{bretto2013hypergraph} for their construction of a hypergraph adjacency matrix. Bretto defines the \textit{adjacency} between a pair of vertices $u,v \in V$, $u \neq v$ as $a_{uv}=|\{e \in E : u,v \in V_e\}|$. Adjacency relationships are identical between corresponding primal and dual elements, so it is natural to consider the adjacency between a pair of hyperedges $e,f \in E$ as being equal to the adjacency of their dual vertices $e',f' \in V'$, i.e., $a_{ef} = a_{e'f'}$.
We define the \textit{adjacency factor} of a vertex $v \in V$ and a hyperedge $e \in E$ as
\begin{equation*}
  \text{Adj}(v) = \sum_{u \in V, \> u \neq v}{a_{uv}^t}, \hspace{0.25in} \text{Adj}(e) = \sum_{f \in E, \> f \neq e}{a_{ef}^t},
\end{equation*} where $t \geq 0$ is used to adjust the influence of vertex pairs with multiple shared hyperedges and hyperedge pairs with multiple shared vertices. Notice that setting $t=0$ simply gives the number of vertices adjacent to $v$, i.e. $|V_v|$. Setting $t>0$ results in a larger adjacency factor for vertices having high adjacency with their neighbors. We discuss the ideal value for $t$ in the next section, in conjugation of the planarity issue of polygon representations of hypergraphs.

\begin{figure}[t]
  \centering
  {\fontfamily{phv}\selectfont {\large \textbf{Primal} \hspace{0.5in} \textbf{Dual} \hspace{0.5in} \textbf{Primal} \hspace{0.5in} \textbf{Dual}}} \\
  \subfloat[][]{\includegraphics[width=0.9in]{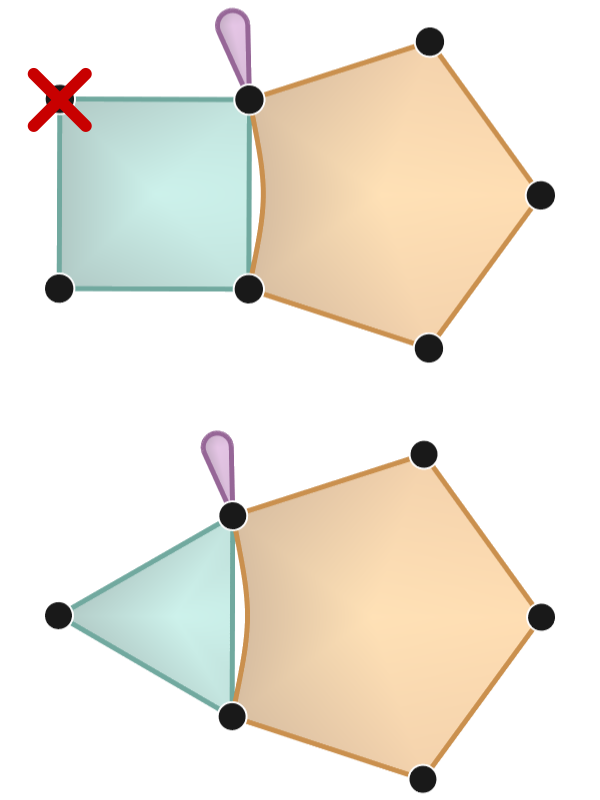}}
  \subfloat[][]{\includegraphics[width=0.9in]{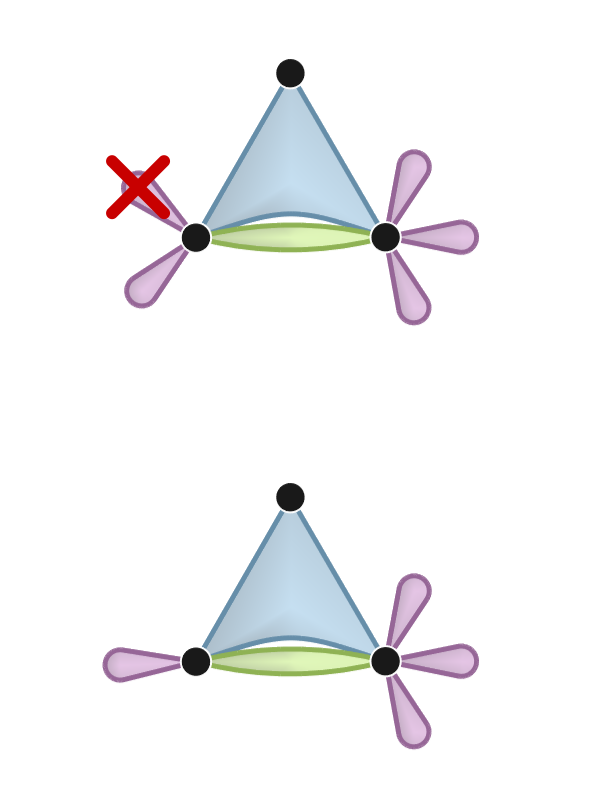}}
  \hspace{-0.13in}\vline
  \subfloat[][]{\includegraphics[width=0.9in]{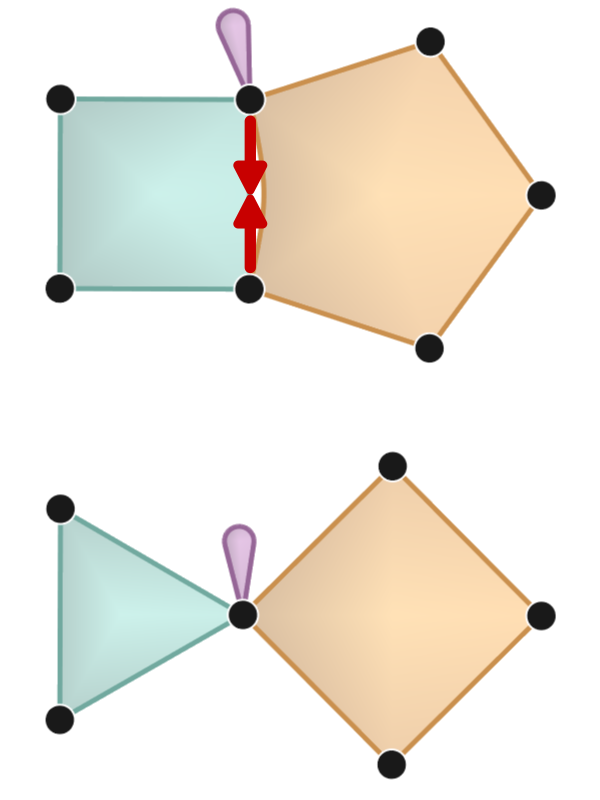}}
  \subfloat[][]{\includegraphics[width=0.9in]{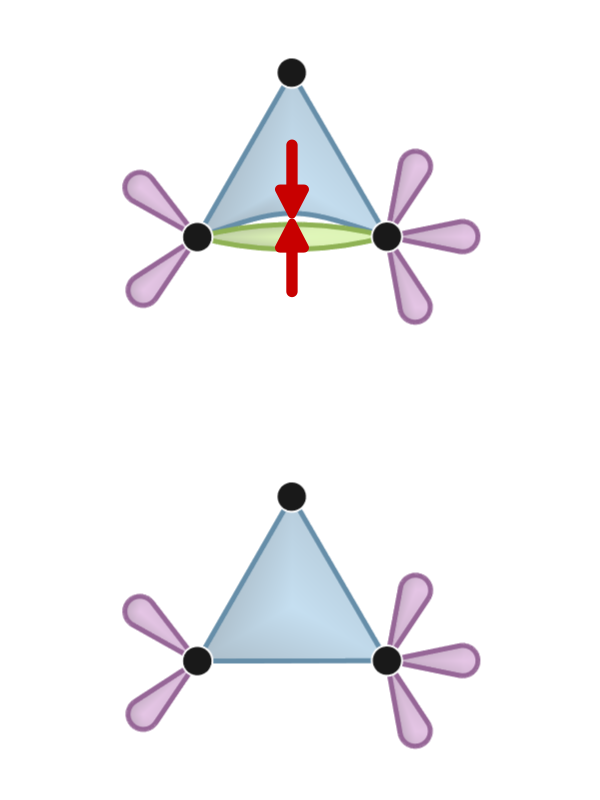}}
  \caption{The four atomic simplification operations. (a) A vertex removal in the primal hypergraph corresponds to (b) a hyperedge removal in the dual hypergraph. (c) A vertex merger in the primal hypergraph corresponds to (d) a hyperedge merger in the dual hypergraph.}\label{fig:operations}
\end{figure}

\subsubsection{Polygon Planarity}
\label{sec:polygon_planarity}

A graph is planar, i.e. an edge crossing-free embedding can be found, if and only if it does not contain a subdivision of the complete graph $K_5$ or complete bipartite graph $K_{3,3}$~\cite{Kuratowski1930}.

Recall that a hypergraph $H$ is Zykov planar if its K{\"o}nig representation $K(H)$ is a planar graph (\Cref{sec:background}). This definition of planarity assumes that hyperedges can be represented as arbitrary closed regions. However, when requiring that the regions be drawn as near-regular polygons, as in \cite{Qu:2017,Qu:22}, Zykov's definition is insufficient. This motivates a new definition for hypergraph planarity for the (near-regular) polygon representation:

\begin{definition}
  A \emph{convex polygon representation} is a drawing of a hypergraph in the plane where each hyperedge is represented as a strictly convex polygon such that the area of intersection between each pair of polygons is zero.
\end{definition}

We say that a hypergraph is \emph{convex polygon planar} if it admits a convex polygon representation. We have identified four \textit{forbidden sub-hypergraphs} that are Zykov planar but lack a convex polygon representation. We begin by defining an \emph{n-adjacent cluster} as the partial hypergraph induced by a set of hyperedges $J \subseteq E$ which contain a set of vertices $X \subseteq V$, $|X| = n \geq 2$, where each hyperedge in $J$ contains all of the vertices in $X$, that is, $v_i \in e_j$ for all $v_i \in X$ and $e_j \in J$. Our first forbidden sub-hypergraph is a 3-adjacent cluster of two hyperedges (\Cref{fig:forbidden} (a)), and the second is a 2-adjacent cluster of three hyperedges (\Cref{fig:forbidden} (b)). Notice that these sub-hypergraphs are a primal-dual pair: if one appears in the primal view, the other appears in the dual view among the corresponding dual elements. Our third forbidden sub-hypergraph is the neighborhood of a vertex $v$ where a proper subset of its incident hyperedges and adjacent vertices form a cycle of length $n \geq 3$ (\Cref{fig:forbidden} (c)). The fourth is the neighborhood of a hyperedge $e$ where a proper subset of its incident vertices and adjacent hyperedges form a cycle of size $n \geq 3$ (\Cref{fig:forbidden} (d)). These sub-hypergraphs also form a primal-dual pair. We refer to these forbidden sub-hypergraphs as containing a strangled vertex or strangled hyperedge respectively.

\begin{theorem} \label{thm:forbidden}
  Let $H$ be a Zykov planar hypergraph. Then $H$ has a convex polygon representation if and only if it does not contain any of the following as a sub-hypergraph:
  \begin{enumerate}[nosep]
    \item[(a)] A 3-adjacent cluster of 2 hyperedges,
    \item[(b)] A 2-adjacent cluster of 3 hyperedges,
    \item[(c)] A strangled vertex,
    \item[(d)] A strangled hyperedge.
  \end{enumerate}
\end{theorem}

We refer the reader to \Cref{apx:planarity} for a proof of \Cref{thm:forbidden}. Since the forbidden sub-hypergraphs form primal-dual pairs, we further claim that a hypergraph $H$ has a convex polygon representation if and only if its dual hypergraph $H'$ has a convex polygon representation.

\begin{figure}[btp]
  \centering
  \subfloat[][]{\includegraphics[height=0.9in]{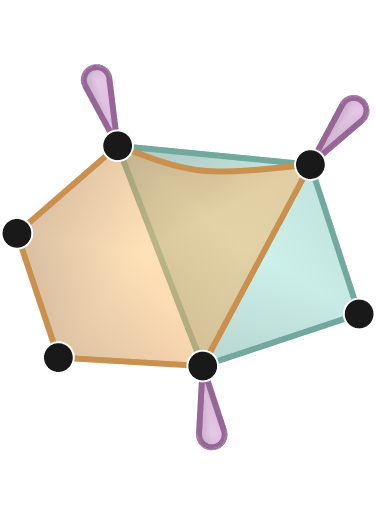}} \hspace{0.2in}
  \subfloat[][]{\includegraphics[height=0.9in]{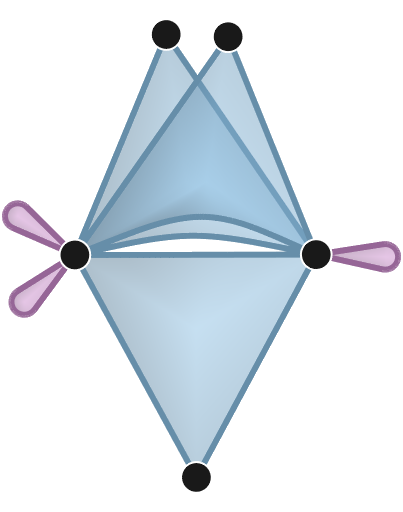}} \hspace{0.2in}
  \subfloat[][]{\includegraphics[height=0.9in]{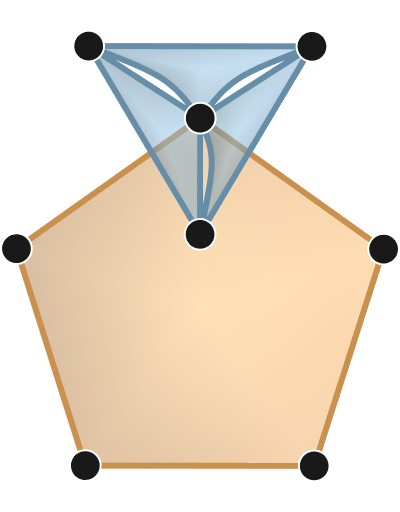}} \hspace{0.2in}
  \subfloat[][]{\includegraphics[height=0.9in]{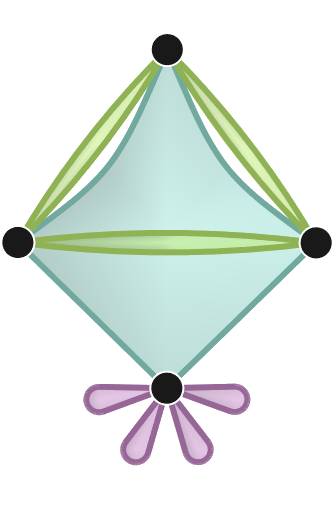}}
  \caption{Examples of the four forbidden sub-hypergraphs in the polygon visualization metaphor: (a) 3-adjacent hyperedge cluster of 2 hyperedges, (b) 2-adjacent hyperedge cluster of 3 hyperedges, (c) strangled vertex, (d) strangled hyperedge. Notice that (b) is the dual of (a) and (d) is the dual of (c).}\label{fig:forbidden}
\end{figure}

We refer to polygon overlaps occurring in a polygon layout of a hypergraph that has a convex polygon representation as {\em avoidable overlaps} (\Cref{fig:entangled}). Otherwise, such overlaps are {\em unavoidable overlaps}. Note that forbidden sub-hypergraphs are the simplest examples of unavoidable polygon overlaps. While a 3-adjacent cluster of 5 hyperedges clearly involves more hyperedge overlaps, it also necessarily contains a 3-adjacent cluster of 2 hyperedges.

Our atomic operations are specifically designed to enable eliminating forbidden sub-hypergraphs. Notice that each of the examples in \Cref{fig:forbidden} can be converted to a hypergraph with a convex polygon representation using a single vertex or hyperedge operation. We find a good correlation between forbidden sub-hypergraphs (\Cref{fig:adjfactor}) and our adjacency factor when $t=2$. Given this correlation, simplifying elements with high adjacency factor can reduce the number and size of non-convex polygon planar sub-hypergraphs.

\section{Scalable Optimization Framework}
\label{sec:technique}

In this section, we detail our multi-scale polygon layout optimization framework which consists of two iterative processes: iterative simplification, and iterative layout refinement. The goal of the simplification process is to construct a sequence of simplified scales from an input hypergraph $H$ such that each successive scale contains fewer areas of potential polygon overlap, either unavoidable overlaps caused by forbidden sub-hypergraphs or avoidable overlaps caused by a lack of space around high-degree vertices and high-cardinality hyperedges. Either type of polygon overlap can lead to challenges in layout optimization and significant visual clutter in the layouts of large hypergraphs, so we address both simultaneously during simplification. Once the sequence of simplified scales $\{H_0,H_1,\dots,H_n\}$ is generated, the goal of the iterative layout refinement process is to produce a high-quality polygon layout for each scale. We achieve this by first optimizing the layout of the coarsest scale $H_n$, then iteratively inverting simplification operations and locally refining the layout of intermediate scales until the original scale is recovered.

\begin{figure}[t]
  \centering
  \includegraphics[height=1.125in]{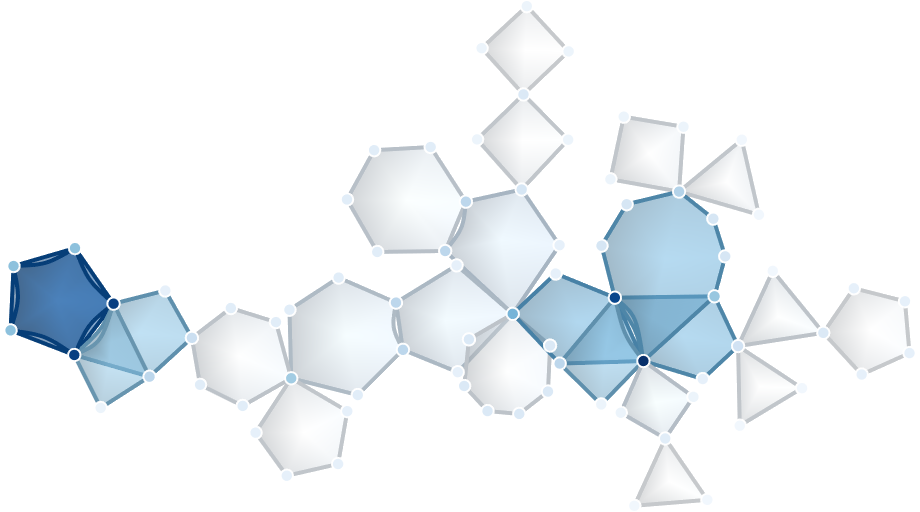}
  \hspace{0.125in}
  \raisebox{0.125in}{\includegraphics[height=0.875in]{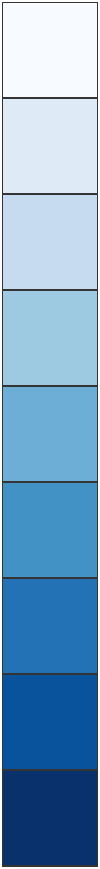}}
  \caption{Hypergraph elements colored according to adjacency factor. The two regions drawn in dark blue indicate elements with high adjacency factor and both correspond to forbidden sub-hypergraphs.}
  \label{fig:adjfactor}
\end{figure}

\begin{figure*}[!t]
  \centering
  \subfloat[][Input hypergraph]{\includegraphics[height=1.25in]{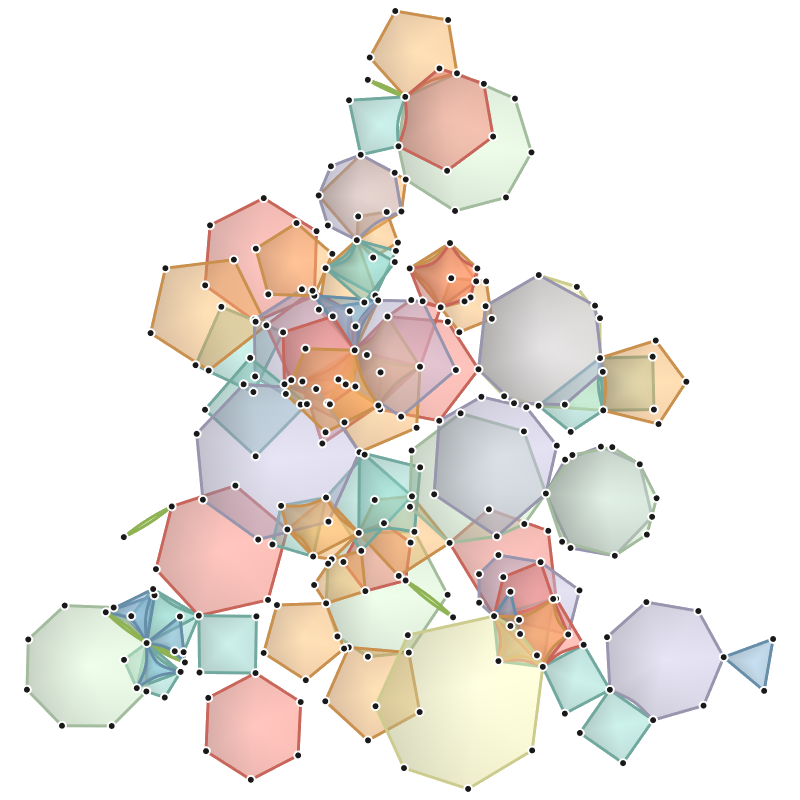}} \hspace{0.125in}
  \subfloat[][$\alpha=1$, $\beta=0$, $\gamma=0$]{\includegraphics[height=1.25in]{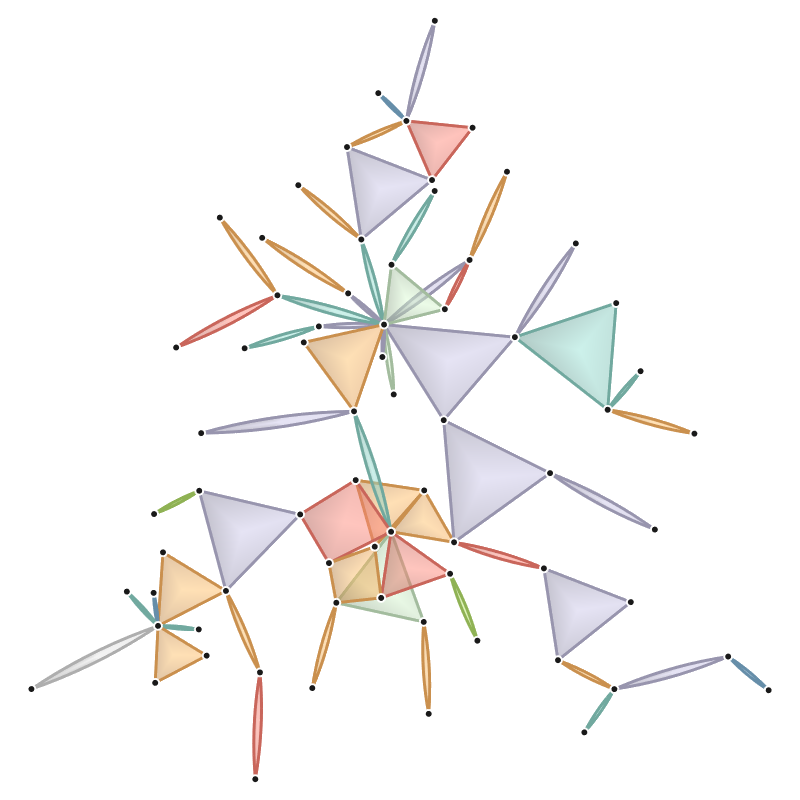}} \hspace{0.125in}
  \subfloat[][$\alpha=0$, $\beta=1$, $\gamma=0$]{\includegraphics[height=1.25in]{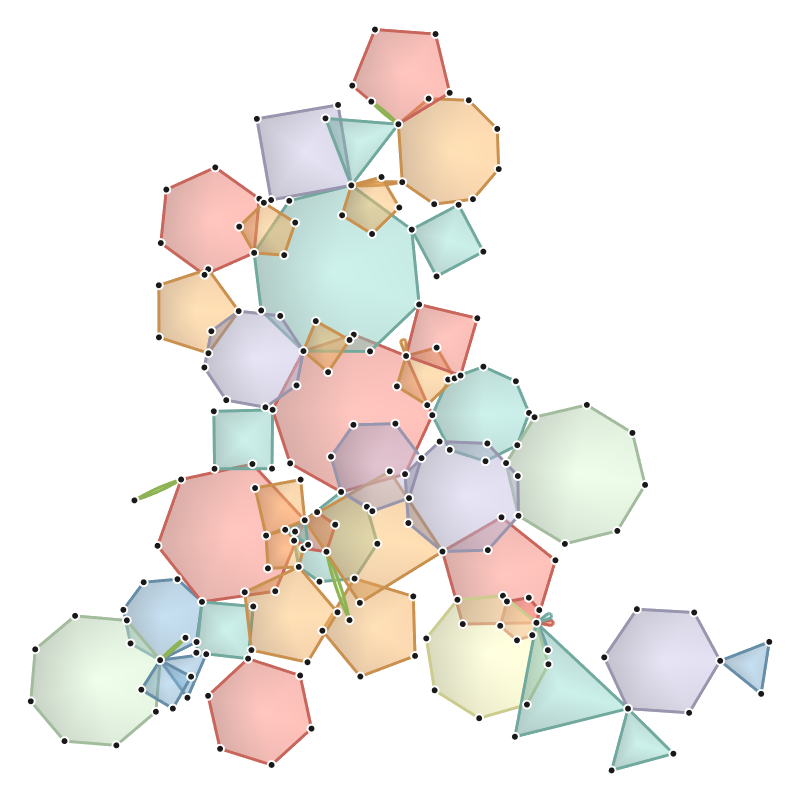}} \hspace{0.125in}
  \subfloat[][$\alpha=0$, $\beta=0$, $\gamma=1$]{\includegraphics[height=1.25in]{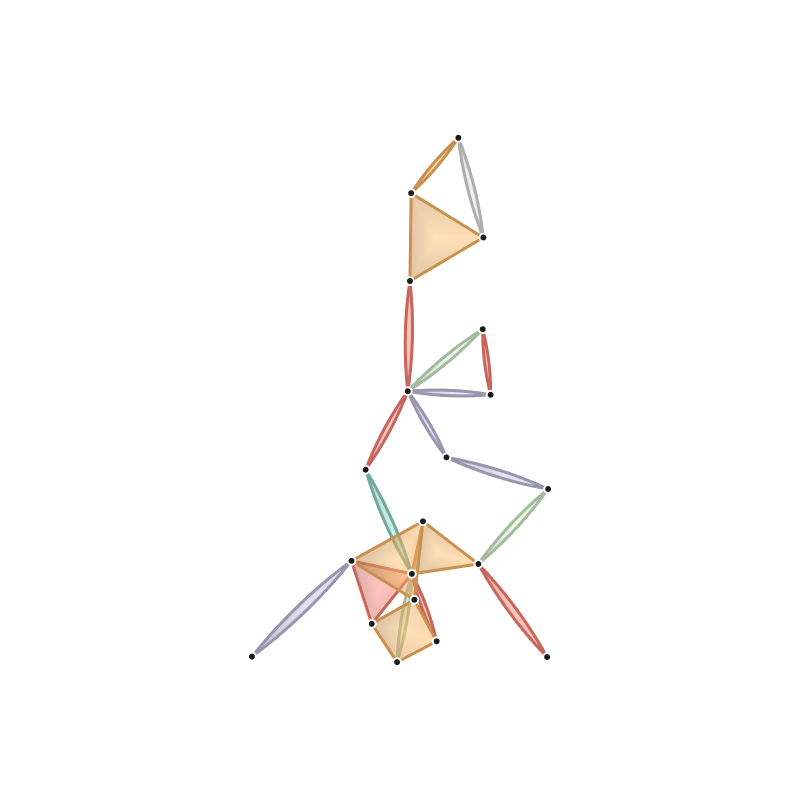}} \hspace{0.125in}
  \subfloat[][${\alpha=0.4,\>\beta=0.4,\>\gamma=0.2}$]{\includegraphics[height=1.25in]{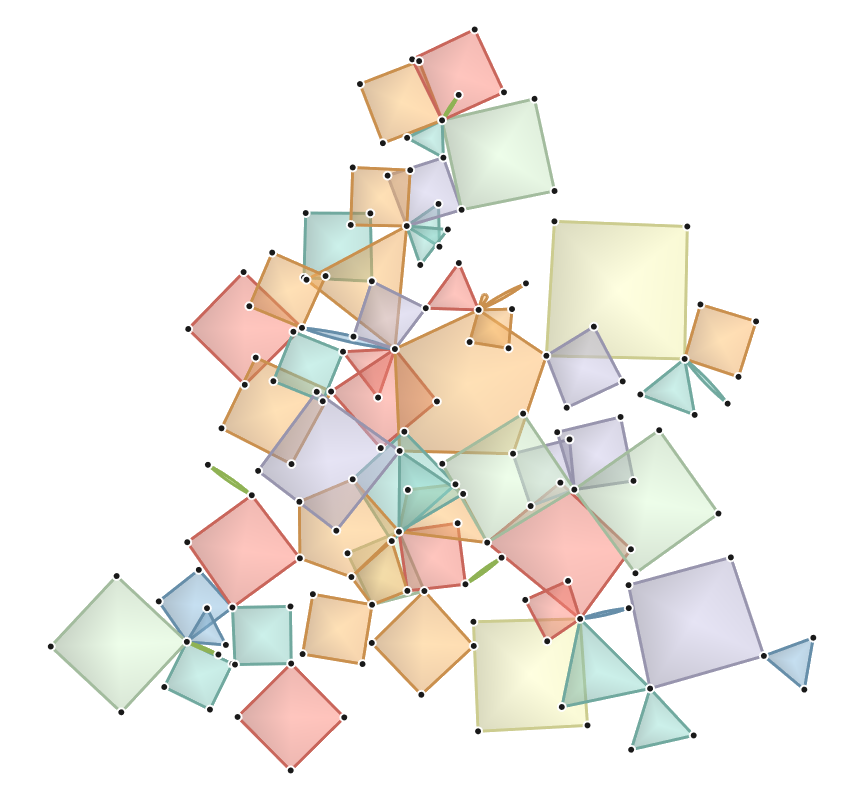}}
  \caption{A paper-author network with $260$ vertices and $83$ hyperedges is simplified with our system using different values for the weight parameters $\alpha,\beta,\gamma$ in \Cref{eq:priority}. In (b), only the term targeting high-degree vertices and high-cardinality hyperedges is used, requiring $220$ atomic operations to make the hypergraph linear. Many of the polygons in this layout have been simplified to digons, making it difficult to determine the relative sizes of the hyperedges in the original data. In (c), only the term targeting elements with high adjacency factor is used, requiring $103$ atomic operations to make the hypergraph linear. The relative sizes of the hyperedges are more accurately maintained, but the visualization still contains many avoidable polygon overlaps, especially around high-degree vertices. In (d), only the term preserving elements with high betweenness centrality is used, requiring $304$ atomic operations to make the hypergraph linear. Most of the information on vertex degree and hyperedge cardinality is lost in this visualization. We can however see evidence of three hypergraph cycles which are not apparent in the other simplifications. In (e), all three terms are used, requiring 135 simplification operations to make the hypergraph linear. This priority weighting scheme and visualization preserves some information on relative hyperedge sizes and also reduces visual clutter around high-degree vertices.}
  \label{fig:weights}
\end{figure*}

\subsection{Simplification Operation Generation}
\label{sec:atomic_operations}

We initially generate removal operations for every vertex and hyperedge in $H_0$, and merger operations between every pair of adjacent vertices and hyperedges. Notice that removing a vertex or hyperedge arbitrarily has the potential to make a hypergraph disconnected. To avoid this, we constrain the \textit{legality} of removal operations. For a hyperedge $e$ at the current hypergraph scale, we mark its removal operation as \textit{illegal} if it would make a pair of vertices $u,v \in e$ non-adjacent. Otherwise, we mark it as \textit{legal}. Similarly, a vertex removal operation is marked illegal if it makes any pair of incident hyperedges non-adjacent. By allowing only legal removal operations, our simplification ensures that the hypergraph remains connected in each simplified scale. We also constrain the legality of merger operations to avoid simplifying portions of the hypergraph that are already linear. For a pair of hyperedges $e,f$ at the current hypergraph scale, we mark their merger operation as legal if their adjacency $a_{ef} \geq 2$ and illegal otherwise. For a pair of vertices $u,v$ in the current hypergraph scale, we mark their merger operation as legal if their adjacency $a_{uv} \geq 2$ and illegal otherwise. After the legality of each generated operation has been determined, we place the legal operations in a priority queue keyed on a simplification priority measure.

Our operation priority measure consists of three terms based on the following statistics: vertex degree (hyperedge cardinality), adjacency factor, and betweenness centrality. As discussed in \Cref{sec:polygon_planarity}, adjacency factor is correlated to the presence of a forbidden sub-hypergraph. Instead of trivially deleting or collapsing forbidden sub-hypergraphs, we use a term based on adjacency factor to promote simplifying sub-hypergraphs until they have a convex polygon representation (\Cref{fig:weights} (c)). We use the term based on vertex degree and hyperedge cardinality to promote reducing the space required by high-degree vertices and large hyperedges in the polygon layouts of simplified scales. This can help to reduce avoidable polygon overlaps (\Cref{fig:weights} (b)). Finally, we use the term based on betweenness centrality to promote preserving centrally located elements that are relevant to the path structure of the input hypergraph (\Cref{fig:weights} (d)). By combining these terms, we are able to generate simplified scales with reduced visual clutter in areas with the most polygon overlaps, while also retaining the relative polygon sizes and core connectivity of the hypergraph (\Cref{fig:weights} (e))

To normalize the distributions of each statistic, we also require the global minimum and maximum values of the input hypergraph $H_0$ (and its dual $H_0'$) for vertex degree and hyperedge cardinality, $d_{min}, d_{max}$, hyperedge adjacency factor, $a_{min}, a_{max}$, and betweenness centrality, $b_{max}, b_{min}$. Given an atomic simplification operation $O$, the final priority measure is a weighted sum of our three terms given by
\begin{align} \label{eq:priority}
  \text{P}(O)\! = \alpha\! \left(\!\frac{\hat{d}_O\!-d_{min}}{d_{max}\!-d_{min}}\right)\! + \beta\! \left(\!\frac{\bar{a}_O\!-a_{min}}{a_{max}\!-a_{min}}\right)\! + \gamma\! \left(\!\frac{b_{max}\!-\bar{b}_O}{b_{max}\!-b_{min}}\right).
\end{align}

Here we use $\bar{a}_O$ to denote the adjacency factor of the removed element in the case of a removal operation, and the average adjacency factor of the merged elements in the case of a merger operation. Similarly, $\bar{b}_O$ denotes the average betweenness centrality of the operand elements. We use $\hat{d}_O$ to denote the maximum vertex degree or hyperedge cardinality in the footprint of $O$ (\Cref{sec:atomic_operations}). That is,
\begin{align} \label{eq:degree}
  \hat{d}_O = \max \biggl\{ \max_{v \in V(H_O)}\bigl\{deg(v)\bigr\}, \max_{e \in E(H_O)}\bigl\{card(e)\bigr\} \biggr\}.
\end{align}
This helps to simplify elements surrounding a high-degree vertex instead of removing or merging the high-degree vertex itself as demonstrated in \Cref{fig:weights} (b).

We compute betweenness centralities once for the input hypergraph $H_0$ and use these values for the entire simplification process. This is because we aim to preserve the path structure of the input hypergraph, not the path structure of the previous simplified scale. This also avoids having to recompute betweenness centrality which can be costly. We do however update vertex degrees, hyperedge cardinalities, and adjacency factors after the application of each operation.

\subsection{Iterative Simplification}
\label{sec:simplification}

At this stage, our priority queue only contains operations that can legally be applied to the input hypergraph $H_0$ with the highest priority operations at the front. Our framework proceeds to iteratively simplify $H_0$ by popping operations from the priority queue and applying them. We apply simplification iteratively to accommodate changes in legality or priority that any operation can incur. With each iteration, the priority queue is re-sorted to keep the highest priority operations at the front. When a legal vertex removal is performed, the vertex is removed from each of its containing hyperedges, reducing their cardinality by one, and subsequently deleted. The process is similar for a hyperedge removal. When a pair of vertices $u,v \in V$ are merged through a legal operation, we first update $v$ to include all the hyperedges incident to $u$ and then remove $u$ from the hypergraph. We call $u$ the \textit{removed} vertex and $v$ the \textit{retained} vertex. The process for hyperedge-based mergers is identical except that the roles of vertices and hyperedges are reversed.

When an operation $O$ is applied, the operands, legality, and priority of operations on elements in the footprint of $O$ may need to be updated. In addition, merger operations have the potential of making previously nonadjacent elements adjacent, making them eligible for merging. After $O$ is applied and the appropriate updates made, a record of the operation is added to a stack data structure containing applied operations eligible for future reversal. When the footprint of each operation is relatively small, the necessary updates can be completed efficiently. If $O$ removes an element $r \in V \cup E$, any merger operations involving $r$ become invalid and are removed from the priority queue. For each element $s \neq r$ in the footprint of $O$, any operation involving $s$ must have its priority updated since the incidence and adjacency relationships within its footprint will have changed. In addition, the legality of the removal operation on $s$ must be revisited since the connectivity within its footprint will have changed. If any previously illegal operation is found to be legal its priority is recalculated and the operation is added to the queue. If any previously legal operation is found to be illegal, it is removed from the priority queue.

If $O$ merges an element $r \in V \cup E$ into $t \in V \cup E$, the removal operation on $r$ becomes invalid and is removed from the priority queue. Similar to a removal operation, for each element $s \notin \{r,t\}$ in the footprint of $O$, any operation involving $s$ must have its priority updated, and if it is a removal operation, its legality must be revisited. Furthermore, if a merger operation exists between $s$ and $r$, the reference to $r$ is replaced with a reference to $t$. Finally, if a pair of elements $s_1,s_2\notin \{r,t\}$ in the footprint of $O$ become newly adjacent, we initialize a new merger operation between them, calculate the priority of the operation, and add it to the priority queue.

We provide several options for defining the coarsest simplified scale $H_n$. The user can set a target number for the vertices or hyperedges in $H_n$, or specify that simplification terminates as soon as the hypergraph becomes linear or is free of forbidden sub-hypergraphs. Once a termination criterion has been met, the stack of applied operations $\{O_n,O_{n-1},\dots,O_1\}$ records the operations in the reverse order of how they were applied.

\subsection{Simplification Reversal and Layout Optimization}
\label{sec:optimization}

Once the coarsest hypergraph scale $H_n$ has been reached, we construct the dual hypergraph $H_n'$ (if used), and apply a modified version of the automatic polygon layout framework of Qu et al.~\cite{Qu:22}. Their objective function for layout optimization includes five energy terms: polygon regularity energy, polygon area energy, polygon separation energy, polygon intersection regularity energy, and primal-dual coordination energy. They minimize this objective function by iteratively adjusting vertex locations in the primal and dual hypergraphs using an L-BFGS quasi-Newton solver~\cite{liu1989limited}. Instead of minimizing each of these energies simultaneously, our modified version starts with a \textit{separation} phase which uses only the polygon separation energy and primal-dual coordination energy, followed by a \textit{regularity} phase that incorporates all five energy terms including the polygon regularity, area, and intersection regularity energies.

The purpose of the separation phase is to unravel any avoidable overlaps present in the initial layout of $H_n$ (i.e. to separate crossing paths or twisted cycles). Qu et al.~\cite{Qu:22} define the separation energy between a pair of polygons as a function of the difference between their current separation and the minimum acceptable separation if the polygons are assumed to be regular. That is, given two polygons $\Gamma_1, \Gamma_2$, the separation energy between them is given by $E_{PS}(\Gamma_1,\Gamma_2) =  f(d(\Gamma_1,\Gamma_2)-d_0(|\Gamma_1|, |\Gamma_2|))$~\cite{Qu:22}. Here $d(\Gamma_1,\Gamma_2)$ is the current distance between polygon centroids and $d_0(|\Gamma_1|, |\Gamma_2|) = \rho_{|\Gamma_1|} + \rho_{|\Gamma_2|} + d_b$ is determined by the circumradii of regular polygons with cardinalities $|\Gamma_1|$ and $|\Gamma_2|$, and a constant buffer distance $d_b$. Since we do not optimize polygon regularity during the separation phase, we alter $d_0$ to be determined by the buffer distance and current radii of $\Gamma_1$ and $\Gamma_2$:
\begin{equation}
  d_0(|\Gamma_1|, |\Gamma_2|) = \frac{1}{2}\left(\max_{u,v \in \Gamma_1}d(u,v) + \max_{u,v \in \Gamma_2}d(u,v)\right) + d_b.
\end{equation}

In the regularity phase of the layout optimization for $H_n$, we re-use the objective function of \cite{Qu:22}. However, instead of using the cardinalities of the hyperedges in $H_n$ for polygon area and separation energy calculations, we use the corresponding cardinalities saved from the original scale $H_0$. This helps to ensure that enough space is reserved for elements that are reintroduced during iterative layout refinement.

Once the layout of the coarsest scale $H_n$ (and $H_n'$) has been optimized, we enter an iterative process in which the applied simplification operations are reversed and the layout is refined. During each iteration, the operation at the top of the stack of recorded operations $O_{i}$ is popped, and its inverse operation $O_{i}^{-1}$ is applied to the current hypergraph scale $H_{i}$. The corresponding inverse operation is applied to the dual of the current hypergraph scale, $H_{i}'$. Whenever a vertex addition is applied to the primal hypergraph, the new vertex is positioned so that it aligns with the center of the corresponding new hyperedge in the dual hypergraph. The same is true when a vertex split is applied to the primal hypergraph: the split vertices are aligned with the centers of their counterparts in the dual hypergraph. We then optimize the positions of vertices inside the footprint of $O$ (in both primal and dual) while keeping all other vertices fixed. The polygon locations in the fixed portion of the layout are still used in the separation energy computation, but the remaining energies are only computed over the footprint of $O$. Confining the layout optimization to the local operation footprint in this way has two benefits: it speeds up the gradient and line search computations used with the L-BFGS solver and helps promote consistent vertex locations between simplified scales.

\section{Case Studies}
\label{sec:casestudies}

We apply our framework to two real-world cases: a network of international trade agreements, and a paper-author collaboration network.

\Cref{fig:trade} shows a network of international regional trade agreements (RTAs) in force as of May 2022 retrieved from the World Trade Organization Regional Trade Agreements Database~\cite{wto:2022}. The dataset excludes bilateral trade agreements as well as trade agreements where one of the parties is itself an RTA. The largest RTA contains 36 participating nations and the largest number of trade agreements that a single nation participates in is 8. In the visualizations, we show both the primal and dual polygon layouts for the original scale (\Cref{fig:trade,fig:trade} (a,c)) and one of the simplified scales (\Cref{fig:trade} (b,d)). In the primal layouts, RTAs are drawn as polygons with incident vertices representing their participating nations. In the dual layouts, the nations are drawn as polygons and trade agreements as vertices. In the bottom left of the original scale primal layout (\Cref{fig:trade} (a)), we see a large pink polygon with many degree-1 vertices. This polygon represents a trade agreement between island nations in the Caribbean with small services-based economies. We observe that many of the other nations participating in only one RTA have relatively small economies. These nations are easier to see in the dual visualization where they are drawn as monogons with a distinctive water-drop shape. In the original scale primal layout (\Cref{fig:trade} (a)), the overlapping polygons in the center of the visualization make it difficult to tell which trade agreement has the most participants. In the dual layout (\Cref{fig:trade} (c)), we can more clearly see a vertex in the center of the visualization with a particularly high degree. This vertex represents the Global System of Trade Preferences among Developing Countries (GSTP) which is the largest RTA in our dataset. In the simplified scale (\Cref{fig:trade} (b,d)), we can see clusters of overlapping polygons in both the primal and dual layouts which indicate different trade blocs. In the upper right of the simplified layouts, we can see a cluster containing nations in Eastern Europe. In both the simplified and original scales, we can observe that this cluster is connected to the rest of the hypergraph through a single path, indicating that the trade bloc in Eastern Europe is somewhat isolated in the global economy.

\begin{figure}[t]
  \centering
  \raisebox{0.375in}{\subfloat[][]{\includegraphics[angle=90,width=1in]{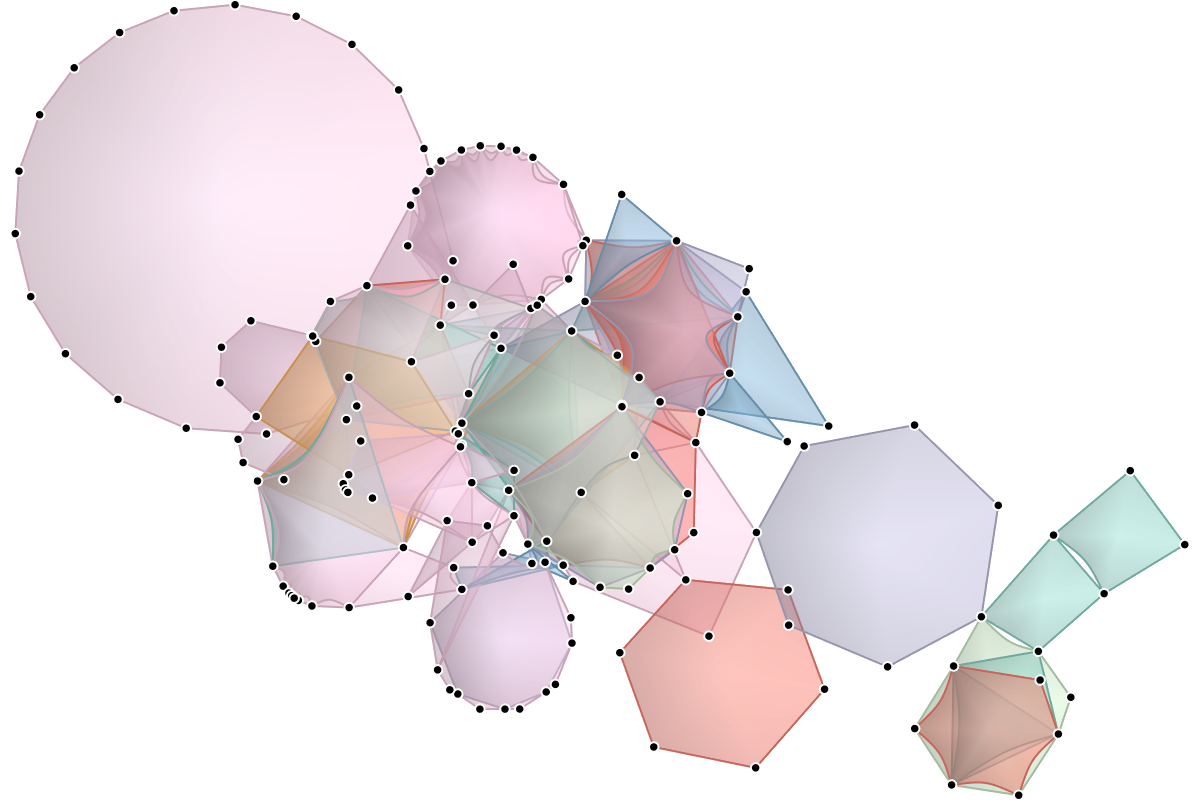}}}
  \hspace{-0.375in}
  \raisebox{0.0in}{\subfloat[][]{\includegraphics[angle=90,width=1in]{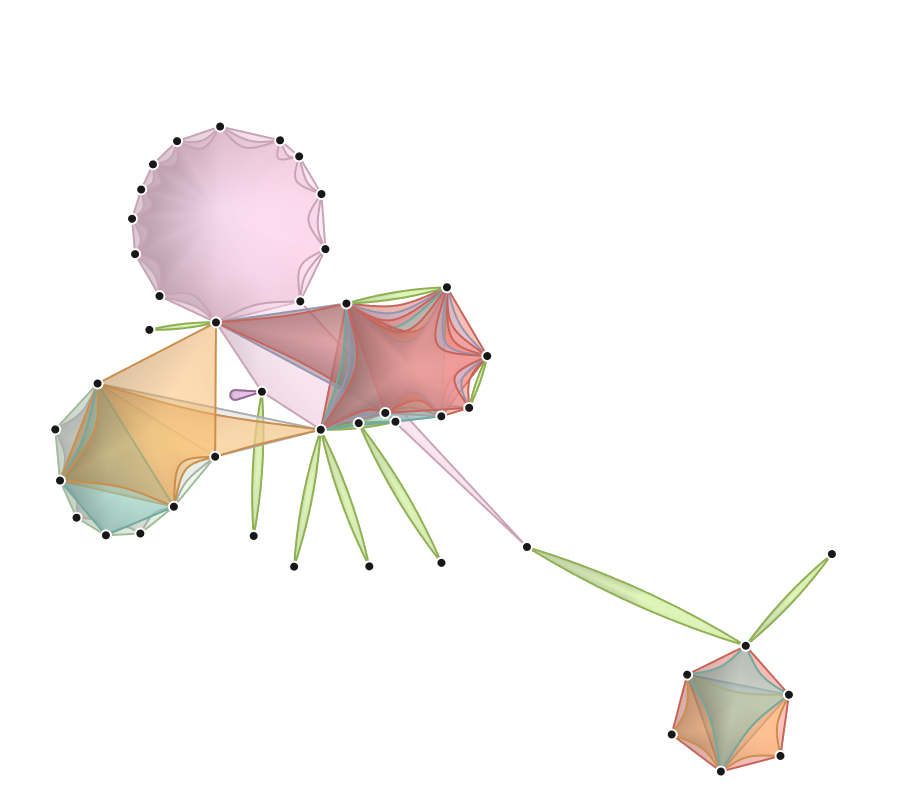}}}
  \raisebox{0.375in}{\subfloat[][]{\includegraphics[angle=90,width=1in]{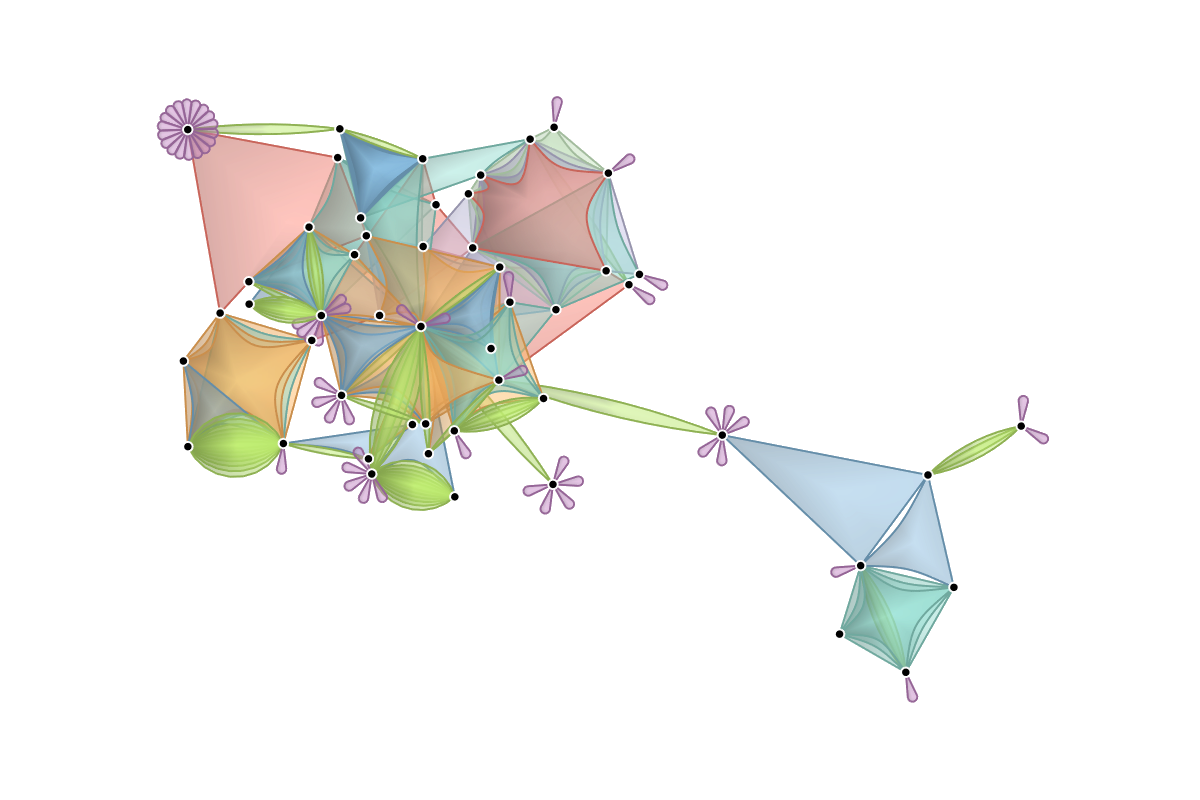}}}
  \hspace{-0.375in}
  \raisebox{0.0in}{\subfloat[][]{\includegraphics[angle=90,width=1in]{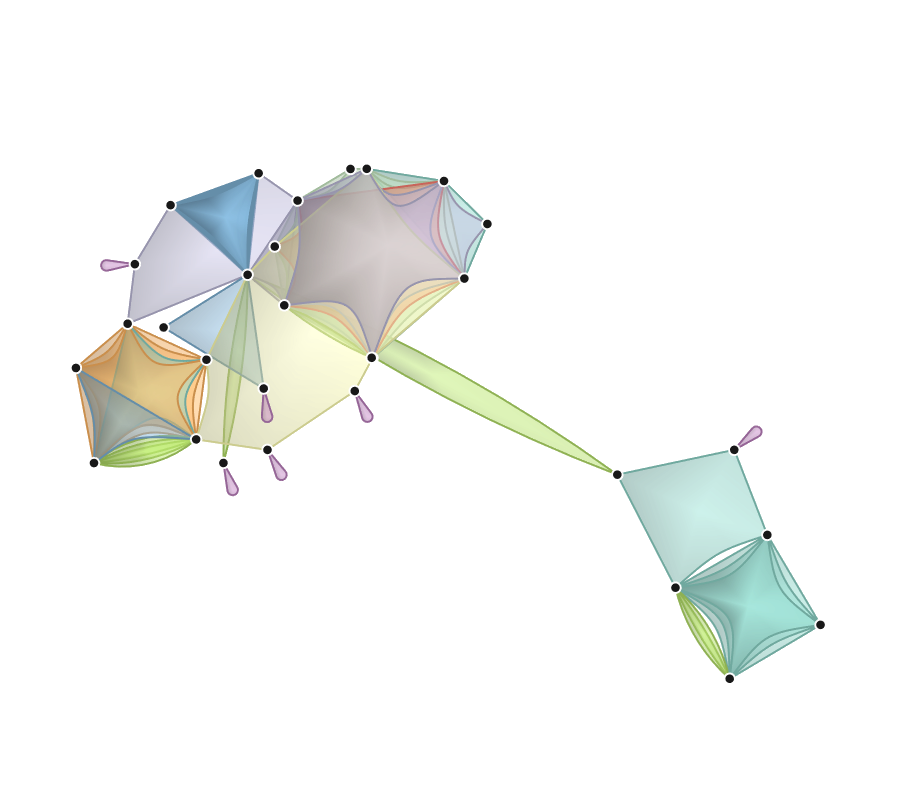}}}
  \caption{Primal and dual visualizations of regional trade agreements and their participating nations. (a) Input primal layout. (b) Simplified primal layout. (c) Input dual layout. (d) Simplified dual layout.}
  \label{fig:trade}
\end{figure}

\begin{figure}[t]
  \centering
  \subfloat[][]{\includegraphics[angle=90,height=2in]{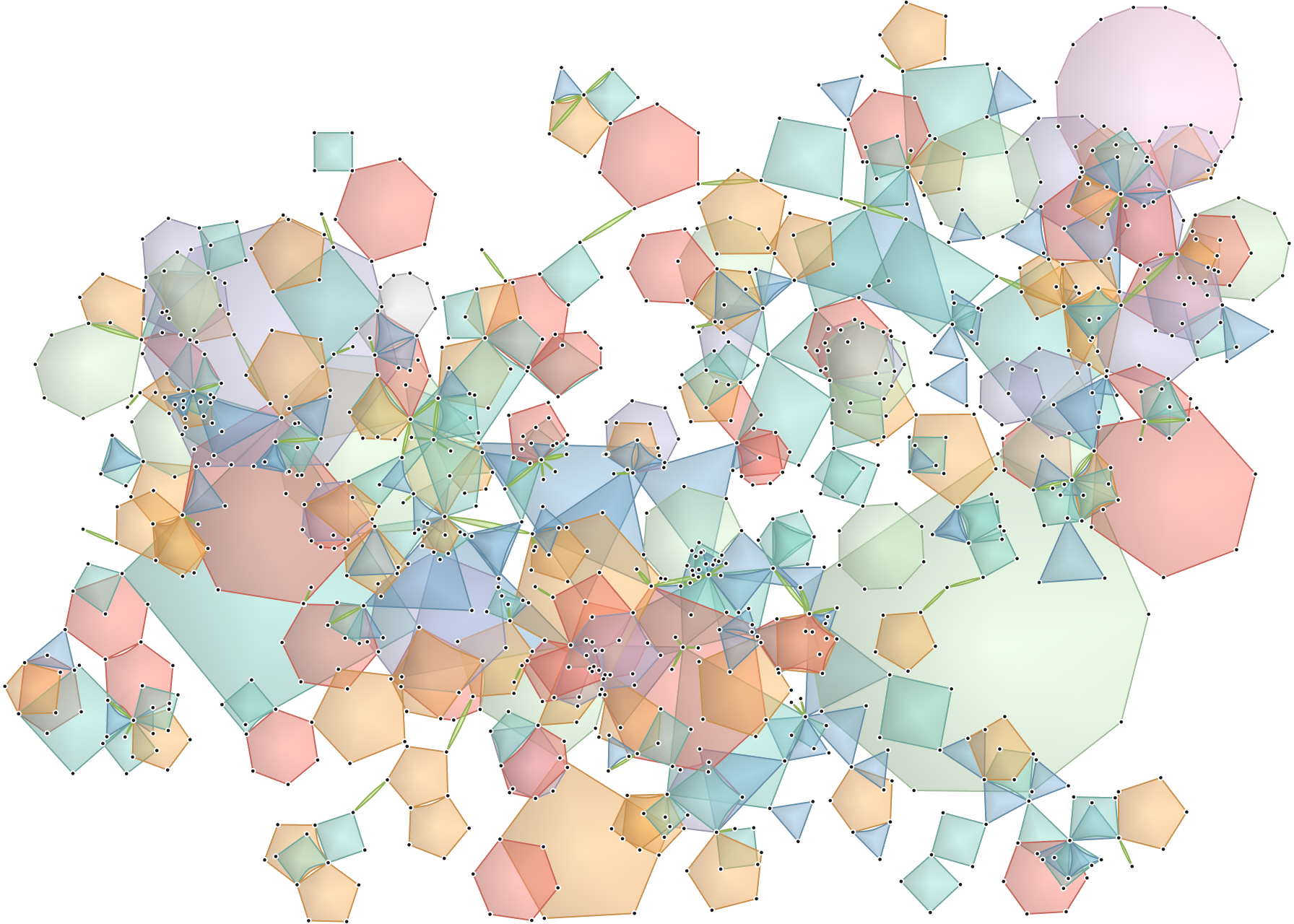}}
  \hspace{0.125in}
  \subfloat[][]{\includegraphics[angle=90,height=2in]{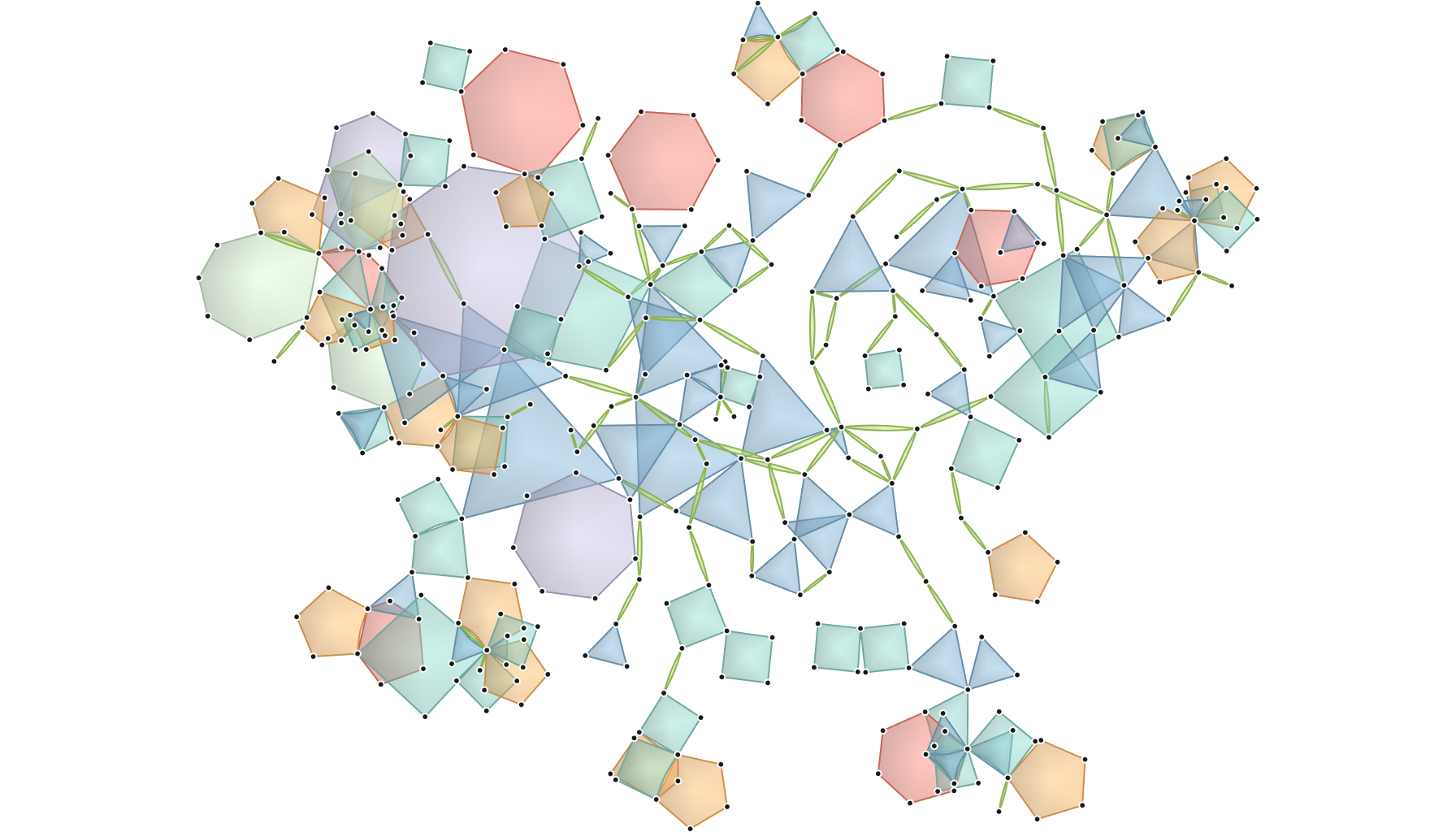}}
  \hspace{0.125in}
  \subfloat[][]{\includegraphics[angle=90,height=2in]{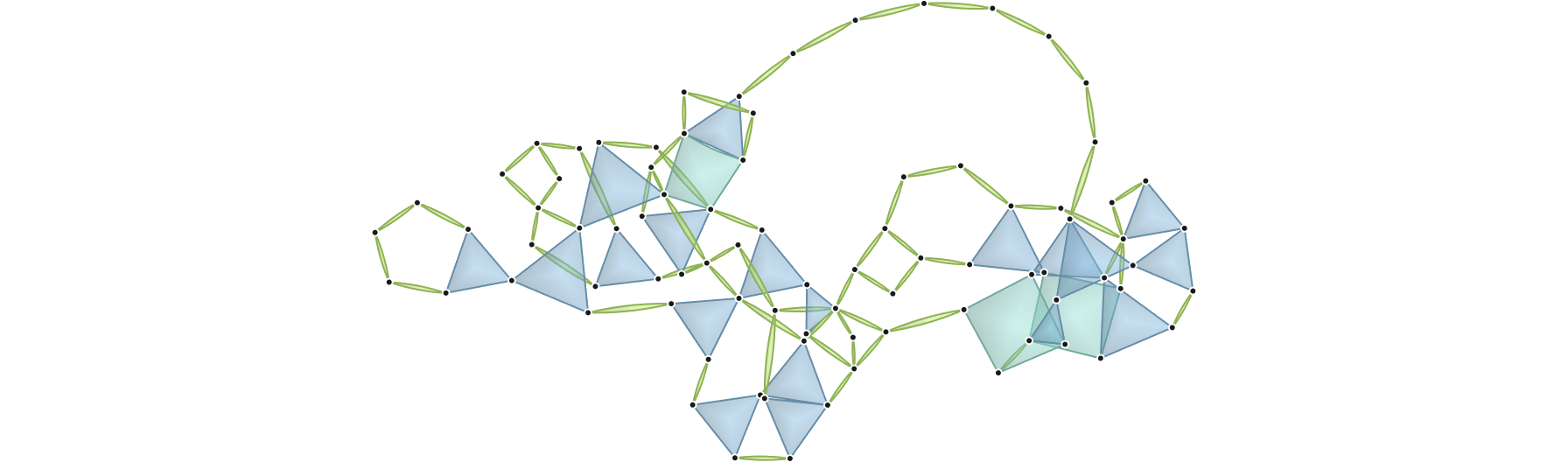}}
  \caption{A paper-author hypergraph dataset containing 1008 vertices and 429 hyperedges (a) is simplified with our framework down to the coarsest allowable scale $H_{1214}$ (c) and the layout is optimized. Then the simplification is iteratively reversed, and the layout is refined at each intermediate scale, an example of which is shown in (b), until the original scale $H_0$ is reached.}\label{fig:paperauthor}
\end{figure}

\Cref{fig:paperauthor} shows the largest connected subset of publications in IEEE Transactions on Visualization and Computer Graphics (TVCG) from $2015$ to $2017$ retrieved from the DBLP database~\cite{author_data}. In these visualizations, each polygon represents a published paper whose incident vertices represent its authors. \Cref{sec:paperauthor} shows an enlarged version of this figure. The dataset contains a total of $1008$ vertices and $429$ hyperedges with a maximum vertex degree of $17$ and a maximum hyperedge cardinality of $20$. \Cref{fig:paperauthor} shows an optimized layout of the original scale (a) as well as the coarsest simplified scale (c) and one intermediate scale (b) generated by our framework. Each of these scales can be used to inspect different aspects of data. In the original scale (\Cref{fig:paperauthor} (a)), we can easily see which papers have the largest number of authors by looking for the largest polygons. We can also see that these papers have numerous authors with only one publication in the dataset. Many such papers concern domain-specific visualization techniques, so we speculate that they may include domain experts as coauthors. For example, the large pink polygon in the top left corner represents a paper on visualizations of concepts from special and general relativity and includes a team of both computer graphics and astrophysics researchers. Another polygon in the bottom left corner represents a paper on interactive visualizations for prostate cancer health risk communication and includes coauthors from computer graphics as well as medical professionals. With the optimized layout of the original scale, we can see that areas with the most polygon overlap contain a forbidden sub-hypergraph or a particularly high degree vertex. In the context of the paper-author network, forbidden sub-hypergraphs like $n$-adjacent clusters indicate a set of authors who collaborate on multiple papers. Such clusters could represent organized research groups that collaborate frequently on a series of papers. High-degree vertices represent authors with many publications who are likely experienced researchers and academic advisors.

In the coarsest simplified scale (\Cref{fig:paperauthor} (c)), we can see several cycle structures among the remaining authors and publications. Where the cycles are small and tangled, we can infer that there is a significant amount of inter-collaboration between the corresponding researchers in the community. Where the cycles are larger and more spread out, there may be less inter-collaboration and more inclusion of experts from other disciplines. In the intermediate simplified scale (\Cref{fig:paperauthor} (b)), we can clearly see branching sub-tree structures along the right side of the visualization. The ends of such branching structures may represent unique subtopics or specific domain applications that are somewhat removed from the main research topics in the journal. For example, the small grouping of polygons at the center-right of the intermediate simplified scale contains the only papers in the dataset studying visualizations of cerebral blood flow in aneurysms.

\section{Eye Tracking Survey}
\label{sec:eyetracking}

We conducted a preliminary user survey among 12 graduate and undergraduate university students to evaluate the usability of our visualizations. We designed the survey to analyze how participants interact with our visualizations when presented with both primal and dual layouts as well as the layouts of several simplified scales. The survey was conducted in person and involved the use of eye-tracking hardware to monitor participants' exploration of the visualizations. The survey included visualizations of our international trade agreement dataset from \Cref{sec:casestudies}, as well as two paper-author collaboration networks retrieved from Isenberg et al.'s openly available Vispubdata dataset~\cite{Isenberg:2017:VMC}. Their dataset contains information on IEEE Visualization publications from 1990-2021. The first of these datasets, shown in \Cref{fig:teaser}, consists of a connected subset of publications containing the keywords ``flow'', ``graph'', and ``machine learning'', while the second consists of the largest connected subset of papers from the InfoVis and SciVis tracks of the IEEE Visualization journal. For each dataset, we asked two questions requiring participants to analyze different properties of specific hypergraph elements. In addition to recording participant responses, we also tracked participant gaze fixation points as they completed each question. We used a Gazepoint GP3 HD 150Hz eye tracking system with a reported accuracy of 0.5-1.0$^{\circ}$~\cite{Gazepoint}. The eye tracking for each survey trial was conducted in a controlled lab environment with a 24-inch 1920x1200p 144Hz monitor.

For the trade agreement dataset, participants were presented with the original scale visualizations of the primal and dual hypergraph layouts generated by our framework (\Cref{fig:trade} (a,c)). The first question asked participants to count the trade agreements involving five or more countries that only participate in one trade agreement. The second question asked participants to count the number of trade agreements that a specific country participates in. Participants answered the first question with $66.7\%$ accuracy and the second question with $83.3\%$ accuracy. Participants who answered the questions correctly spent an average of $58.1\%$ viewing the dual while participants who answered incorrectly spent an average of $22.3\%$ of their time viewing the dual (\Cref{fig:gazetrajectory}). This points to the value of including both the primal and dual layouts for simple analysis tasks.

\begin{figure}[tbp]
  \centering
  \includegraphics[width=1.7in]{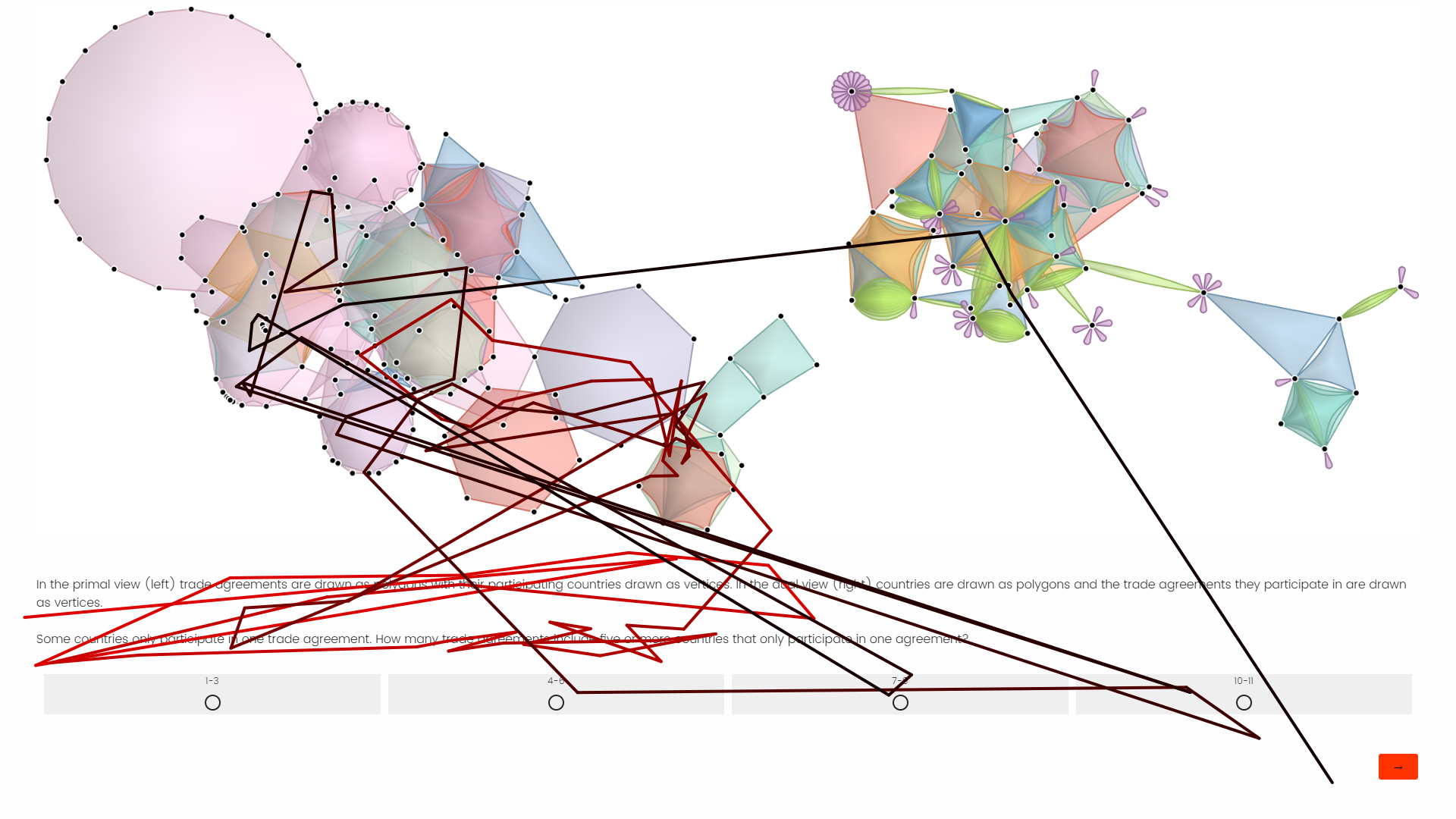}
  \hspace{0.0in}
  \includegraphics[width=1.7in]{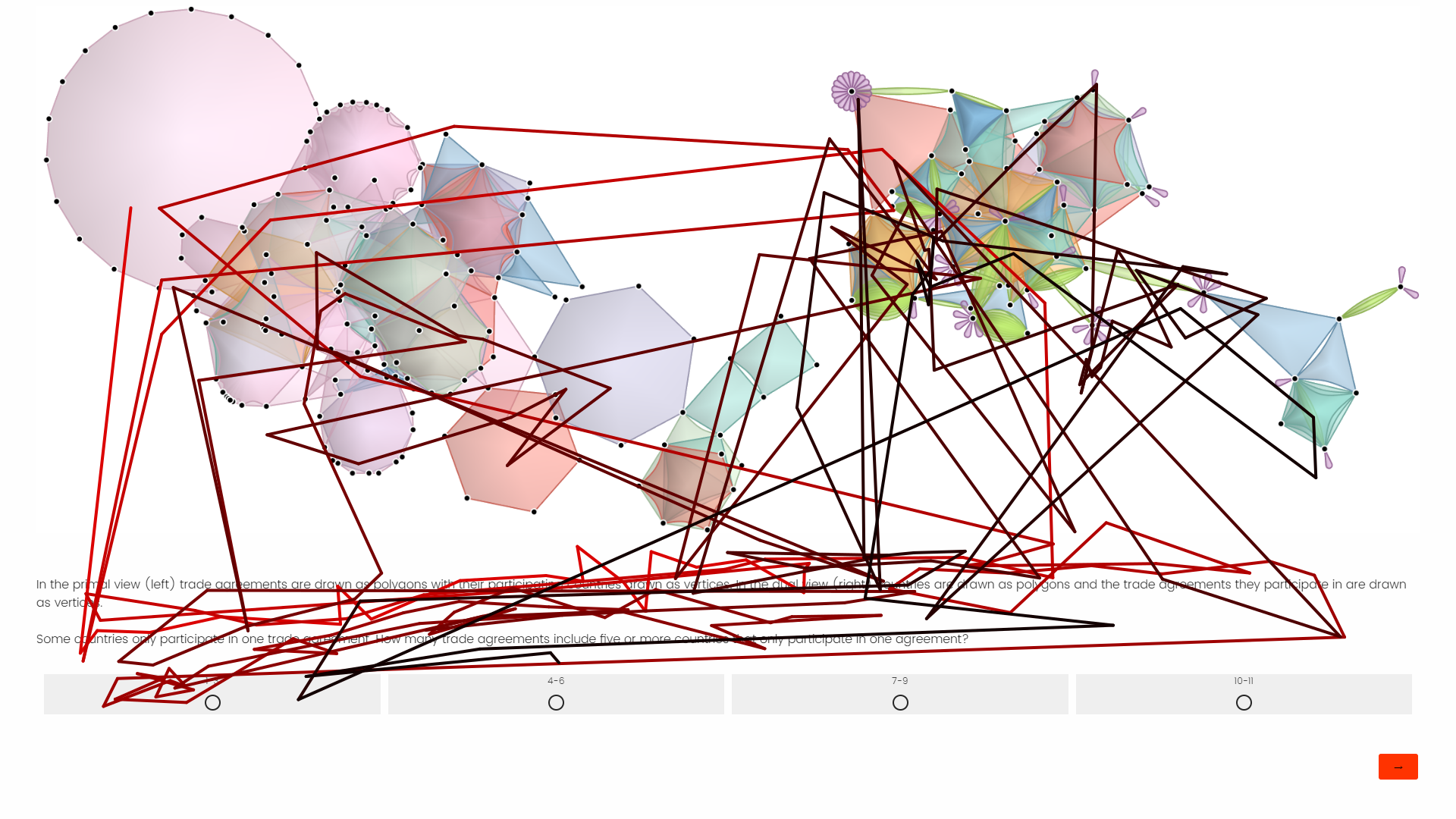}
  \caption{Gaze fixation paths of two participants answering the same question for the trade agreement dataset in our user survey. The participant with the gaze path on the left did not study the dual hypergraph and answered the question incorrectly. The participant with the gaze path on the right studied both the primal and dual hypergraph visualizations and answered the question correctly. Enlarged versions of these plots are available in \Cref{sec:gazetrajectory}.}
  \label{fig:gazetrajectory}
\end{figure}

For the first paper-author dataset (containing $786$ vertices and $318$ hyperedges), participants were presented with both primal and dual layouts of the original hypergraph scale and two simplified scales generated by our framework. The first question for this dataset asked participants to count polygons of a particular color incident to a specific high-degree vertex. The second question asked them to count the number of distinct paths between a pair of vertices. Participants answered the first question with $91.7\%$ accuracy and the second question with $50.0\%$ accuracy. However, not all participants used the simplified scale layouts. Participants who used the simplified scales answered the first question with $100\%$ accuracy and the second question with $66.7\%$ accuracy. Participants who did not use the simplified scales answered the first question with $75.0\%$ accuracy and the second question with 0\% accuracy. On average, participants spent less time viewing the simplified scales than the original scale, however, the participants did not necessarily find the original scale visualization more useful. For example, \Cref{fig:gazetimeline} shows a participant who spent most of their time studying the original scale layouts but was able to arrive at the correct answer soon after viewing the second simplified scale.

\begin{figure}[tbp]
  \centering
  \vspace{12pt}
  \includegraphics[width=3.5in]{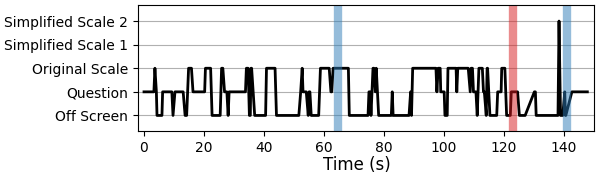}
  \caption{Gaze fixation timeline of a participant answering a question in our user survey. The vertical axis indicates different regions on the participant's screen, including the question text and visualization scales. The horizontal axis represents the time in seconds that a participant spent on the question. The vertical lines in the plot indicate when the participant selected an answer. The blue lines indicate a correct answer and the red lines indicate an incorrect answer. An enlarged version of this plot is available in \Cref{sec:gazetimeline}.}
  \label{fig:gazetimeline}
\end{figure}

\begin{table*}[!t]
  \caption{Comparison of polygon layout optimization methods on paper-author datasets. Each dataset is collected from the DBLP database~\cite{author_data} and consists of a maximal connected subset of publications in the specified year range from the given IEEE journals: Transactions on Robotics (TOR), Transactions on Pattern Analysis and Machine Intelligence (TPAMI), Transactions on Visualization and Computer Graphics (TVCG), Transactions on Human-Machine Systems (THMS), Transactions on Learning Technology (TLT), Transactions on Education (TOE), Transactions on Haptics (TOH), Transactions on Cybernetics (TOC). The first six datasets are the same as those used in \cite{Qu:22}. Note that our method leads to much improvement in terms of fewer avoidable overlaps than the method of Qu et al.~\cite{Qu:22}.}
  \vspace{-8pt}
  \begin{threeparttable}
    \scriptsize
    \setlength{\tabcolsep}{4.7pt}
    \centering
    \begin{tabular}{c c c c p{0.0in} c >{\columncolor[rgb]{0.9,0.9,0.9}}c >{\columncolor[rgb]{0.8,0.8,0.8}}c p{0.0in} c >{\columncolor[rgb]{0.9,0.9,0.9}}c >{\columncolor[rgb]{0.8,0.8,0.8}}c}
      \toprule
      \multicolumn{4}{c}{\textbf{Dataset}} & & \multicolumn{3}{c}{\textbf{Qu et al.}\cite{Qu:22}} & & \multicolumn{3}{c}{\textbf{Ours}} \\
      \cmidrule{1-4}\cmidrule{6-8}\cmidrule{10-12}
      & & & forbidden & & execution & pairwise & sum pairwise & & execution & pairwise & sum pairwise \\
      description & $|V|$ & $|H|$ & sub-hypergraphs & & time (s) & overlap count & overlap area & & time (s) & overlap count & overlap area \\
      \hline
        \rule{0pt}{1.25EM}
      TOR (2015-2020) & 22 & 11 & 1 &         & 0.05 & 4 & 1.48 &         & 0.02 & 1 & 0.42 \\
      TPAMI (2015) & 77 & 24 & 6 &            & 0.42 & 60 & 29.11 &       & 0.43 & 15 & 5.63 \\
      TPAMI (2013-2014) & 93 & 42 & 2 &       & 0.93 & 61 & 7.70 &        & 0.56 & 13 & 2.12 \\
      TOR (2015-2020) & 146 & 56 & 8 &        & 2.48 & 92 & 25.29 &       & 2.09 & 30 & 7.88 \\
      TPAMI (2013-2015) & 314 & 126 & 12 &    & 20.15 & 156 & 37.01 &     & 14.17 & 62 & 15.18 \\
      TOR (2013-2020) & 527 & 232 & 37 &      & 61.08 & 422 & 98.12 &    & 63.07 & 154 & 37.12 \\
      \hline
        \rule{0pt}{1.25EM}
      TVCG (2015-2017)& 1008 & 429 & 113\tnote{$\dagger$} &                   & 45.19\tnote{*} & 6366 & 3806.36 &       & 687.14 & 3119 & 1297.50 \\
      THMS, TLT, TOE, TOH & & & & & & & & & & & \\
      (2013-2023) & \multirow{-2}{*}{1754} & \multirow{-2}{*}{635} & \multirow{-2}{*}{265\tnote{$\dagger$}} &   & \multirow{-2}{*}{1917.49\tnote{*}} & \multirow{-2}{*}{6125} & \multirow{-2}{*}{4598.78} &     & \multirow{-2}{*}{4570.77} & \multirow{-2}{*}{2895} & \multirow{-2}{*}{1154.67} \\
      TPAMI (2013-2020) & 2054 & 947 & 364\tnote{$\dagger$} &                 & 219.92\tnote{*} & 22113 & 17331.60 &    & 8583.15 & 8877 & 3666.70 \\
      TOC (2022) & 3047 & 1000 & 172\tnote{$\dagger$} &                       & 4518.35 & 5566 & 3497.09 &              & 13156.90 & 2510 & 863.64 \\
      \bottomrule
    \end{tabular}
    \begin{tablenotes}
      \item[$\dagger$] Datasets also contain unavoidable overlaps due to $K_5$ and $K_{3,3}$ sub-hypergraphs.
      \item[*] Layout optimization did not converge and terminated early.
    \end{tablenotes}
  \end{threeparttable}
  \label{tab:comparison}
\end{table*}

For the second paper-author dataset (containing $1878$ vertices and $966$ hyperedges), participants were presented with the primal layout of the original hypergraph scale and two simplified scales generated by our framework. We did not include the dual layouts because the primal layout required the entirety of the user's screen to be viewed clearly. The first question asked participants to identify the highest degree vertex in the visualization, and the second question asked them to count the length of the shortest path between two hyperedges. Participants answered the first question with $91.7\%$ accuracy and the second question with $66.7\%$ accuracy. All participants used the simplified scale layouts for both questions. On average, participants spent $21.3\%$ of their time viewing the original scale layout, $27.3\%$ of their time viewing the first simplified scale, and $51.4\%$ of their time viewing the second simplified scale. The increased time spent viewing the simplified scales along with a combined accuracy of $79.2\%$ suggests that participants were able to effectively use the simplified scale visualizations.

\section{Performance Evaluation}
\label{sec:performance}

Our framework leverages simplification operations that generally have small, localized footprints. However, if the input hypergraph is complete, where all the vertices are adjacent to each other, the local footprint of an operation can include the entire hypergraph. This represents the worst-case scenario for the computational complexity of our framework. For a hypergraph $H = \langle V,E\rangle$, recall that $n=|V|$ and $m=|E|$. Our framework first identifies all possible removal operations for vertices and hyperedges, and all merger operations between pairs of adjacent elements, requiring $O((n+m)^2)$ time and generating $O((n+m)^2)$ operations. Our operation priority function (\Cref{eq:priority}) computes a maximum over the footprint of each operation. In the worst case, the footprint of an operation contains $n+m$ elements, so ranking the operations by priority requires $O((n+m)^3)$ time for computation and $O((n+m)^2)\log((n+m)^2))$ time for sorting. Iterative simplification also requires updating the local footprint of each applied operation, adding $O((n+m)^3)$ complexity. Our notation here uses $n$ and $m$ from the original hypergraph scale even though the number of vertices and hyperedges is reduced by a constant factor in each iteration. Altogether, our framework's simplification process requires $O((n+m)^3)$ time. However, the running time of our framework is dominated by the iterative optimization process. Qu et al.~\cite{Qu:22} report that their layout optimization method has a lower bound $\omega((n+m)^4)$. Since our iterative optimization process uses a modified version of their method, we also have a lower bound of $\omega((n+m)^4)$. Each iterative layout refinement then has a lower bound of this form relative to the size of the corresponding operation footprint. Hypergraphs for most practical use cases are far from complete, so operation footprints are smaller than $n+m$. Our testing on real datasets, shown in \Cref{tab:comparison}, indicates that the execution time of our full framework is less than the theoretical bounds. \Cref{tab:comparison} compares our method to that of \cite{Qu:22} on datasets with tens to thousands of elements. These datasets also vary in complexity with respect to the number of unavoidable overlaps they contain. We measure this complexity using the number of forbidden sub-hypergraphs in each dataset. For each method, we display execution time as well as the number of pairwise polygon overlaps and sum of pairwise overlap areas in the optimized layouts.

\section{Conclusion and Future Work}
\label{sec:conclusion}

The main contribution of our work is a multi-scale framework for producing high-quality polygon visualizations of hypergraphs. Our framework features a novel top-down iterative simplification process followed by a bottom-up layout optimization process. To our knowledge, it is the first time that hypergraph simplification has been used specifically for layout optimization. Unlike previous work which focuses on either hyperedge-based sparsification or vertex-based coarsening, we introduce a set of atomic hypergraph simplification operations including both hyperedge and vertex-based operations. Our simplification process is guided by a custom operation priority measure which includes terms for multiple objectives: reducing visual clutter around high-degree vertices, eliminating non-planar sub-hypergraphs, and preserving hypergraph path structures. Additionally, our system is designed to handle primal and dual hypergraphs simultaneously and maintain consistency between them throughout the layout optimization process. We also introduce new and necessary theory on planarity for convex polygon drawings of hypergraphs. This includes a new criterion for convex polygon representations akin to Kuratowski's Theorem for planar graphs~\cite{Kuratowski1930}.

A major challenge that remains for our layout optimization framework is its time complexity. While our framework can handle large datasets, the execution time is constrained by an $\omega((n+m)^4)$ lower bound from the quasi-Newton optimization solver. We plan to continue exploring techniques to enhance the speed of our technique including different possibilities for layout initialization, pre-processing, and employing more efficient optimization solvers.

Our simplification system is designed for hypergraphs that are Zykov planar but lack a convex polygon representation according to our definition. We do not directly address the large class of hypergraphs that are non-planar because they contain structures analogous to $K_5$ and $K_{3,3}$. Furthermore, while our atomic simplification operations can be constrained to ensure the preservation of local connections, they do not consider global topological structures in the hypergraph. We plan to investigate topology-aware simplification methods to handle a larger class of non-planar hypergraphs and create multi-scale representations of hypergraphs that preserve their topological properties.

Preliminary results from our user survey and case studies suggest that the simplified scales used in our layout optimization process can also be used for pattern recognition in hypergraph visualizations. As such, we also plan to pursue multi-scale hypergraph representations that are focused on preserving visual structures in simplified scales. We hope to apply such a visualization system to domains like biology and medicine where interaction networks play an important role, and engineering simulation ensembles where numerous intertwined parameters influence simulation results.

\vspace{5pt}
\acknowledgments{The authors appreciate the constructive feedback from our anonymous reviewers.}

\newpage
\bibliographystyle{abbrv-doi-hyperref}
\bibliography{Scalable_Hypergraph_Paper}

\appendix

\clearpage

\section{Convex Polygon Planarity}
\label{apx:planarity}

Here we present a proof of our planarity criterion for convex polygon drawings of hypergraphs (Theorem~\ref{thm:forbidden} from the main paper), which is inspired by Kuratowski's Theorem that states that a graph is planar if and only if it does not contain a subdivision of the complete graph $K_5$ or complete bipartite graph $K_{3,3}$~\cite{Kuratowski1930}. Similarly, Theorem~\ref{thm:forbidden} states that a hypergraph has a convex polygon representation if and only if it does not contain one of our forbidden sub-hypergraphs described in \Cref{sec:polygon_planarity} (\Cref{fig:forbidden}) of the main paper. 

Our proof requires three intermediate results: (1) the definition of a new graph representation corresponding to some polygon drawing of the hypergraph which we call the \textit{face triangulation graph}, (2) a verification that the face triangulation graph meets the criteria of a convex representation provided by Thomassen~\cite{thomassen1984plane}, and (3) a proof that any articulation vertices in the hypergraph must appear on some face boundary of a convex polygon representation.

We first define connectedness and articulation vertices for hypergraphs as well as facial cycles in graph drawings. A graph is connected if there exists a path between every pair of distinct vertices. Connectedness for hypergraphs is defined similarly (see Bretto~\cite{bretto2013hypergraph}). A graph is said to be \textit{biconnected} if it does not contain any articulation vertices. An articulation vertex (also called cut-vertex) is a vertex whose removal makes the graph disconnected. We define biconnected hypergraphs and articulation vertices for hypergraphs in the same way. From here, it is natural to consider a biconnected component of a hypergraph: a maximal set of vertices $X \in V$ such that the sub-hypergraph induced by $X$ is biconnected. Note that the graph consisting of a single edge and its two endpoint vertices is considered a biconnected graph. Similarly, we consider any hypergraph containing a single hyperedge and its incident vertices to be a biconnected hypergraph. Biconnected graphs and hypergraphs are central to our definition of the face triangulation graph and our first criterion for hypergraph convex polygon planarity. A \textit{face} in a planar drawing of a graph is a region in the plane bounded by a set of vertices and edges. The unbounded region outside of the planar drawing is counted as the \textit{exterior face}. The boundary of each interior face defines an \textit{interior facial cycle}, and the exterior face defines the \textit{exterior facial cycle}. We similarly define a face in a planar polygon drawing of a hypergraph as a region in the plane not covered by a hyperedge that is bounded by a set of vertices and polygon sides.

Our first claim regarding convex polygon representations requires a new definition for the \textit{face triangulation graph} of a polygon drawing for a biconnected hypergraph. Let $H=\langle V,E \rangle$ be a biconnected, Zykov planar hypergraph. Then the K{\"o}nig graph $K(H)=(X,Y,D)$ is a planar graph (\Cref{sec:polygon_planarity}). Recall that $K(H)$ is a bipartite graph containing a vertex $x \in X$ for each hypergraph vertex $v \in V$ and a vertex ${y \in Y}$ for each hyperedge $e \in E$ where $(x,y) \in D$ if the corresponding hypergraph elements $v$ and $e$ are incident in $H$. Let $H$ have a polygon drawing determined by some planar representation of $K(H)$ where each vertex $v \in V$ has the same location in the plane as its corresponding vertex $x \in X$ (\Cref{fig:triangulation} (a,b)). Let the vertices of each hyperedge polygon be ordered according to their angular coordinates relative to the corresponding vertex $y \in Y$. With this polygon drawing of $H$, the \textit{face triangulation graph} $T(H)$ is constructed by the following procedure:

\begin{procedure}{\ \\}
  \begin{enumerate}[nosep]
    \item Let $T(H)$ include all the vertices of $H$.
    \item For two vertices $u,v \in V$, let $(u,v)$ be an edge in $T(H)$ if $(u,v)$ form the side of a polygon in the drawing of $H$.
    \item Let $\{F_1,F_2,\dots\}$ be the interior facial cycles in our current construction of $T(H)$ that correspond to a hypergraph face in the drawing of $H$ (\Cref{fig:triangulation} (b)).
    \item For each interior facial cycle $F_i$, add a vertex $c_i$ located in the interior of $F_i$ and edges $(c_i,v)$ for each vertex $v \in V(F_i)$ (\Cref{fig:triangulation}~(c)).
  \end{enumerate}
  \label{pro:triangulation}
\end{procedure}

With this definition, our goal is to show that the face triangulation graph has a planar drawing with convex facial cycles only for a specific class of hypergraphs. Thomassen~\cite{thomassen1984plane} provides a characterization for such graph drawings which they term convex representations.

\newpage

\begin{theorem}[Thomassen~\cite{thomassen1984plane}]
  Let $G$ be a biconnected planar graph and let $S$ be a cycle which is the face boundary of some plane representation of $G$. Let $\Sigma$ be a convex polygon representing $S$. Then $\Sigma$ can be extended into a convex representation of $G$ if and only if
\begin{enumerate}[nosep]
  \item[(i)] each vertex $x$ in $G-V(S)$ of degree at least 3 is joined to $S$ by three paths that are disjoint except for $x$,
  \item[(ii)] each cycle which is edge-disjoint from $S$ has at least three vertices of degree at least 3, and
  \item[(iii)] no $S$-component has all its vertices of attachment on a path of $S$ corresponding to a straight line segment of $\Sigma$.
\end{enumerate}
\end{theorem}

\begin{figure}[tbp]
  \centering
  \subfloat[][Drawing of K{\"o}nig graph $K(H)$.]{\includegraphics[width=2in]{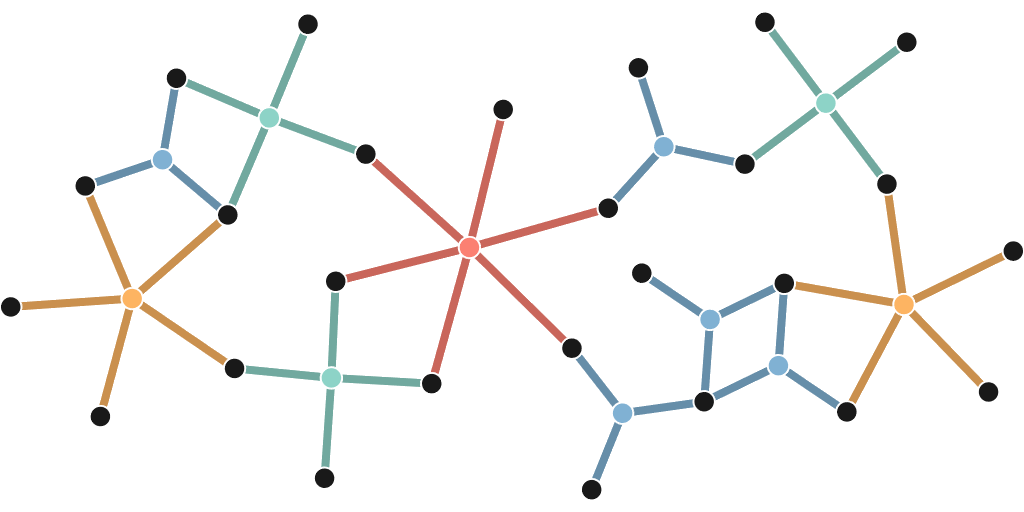}} \\
  \subfloat[][Drawing of hypergraph $H$.]{\includegraphics[width=2in]{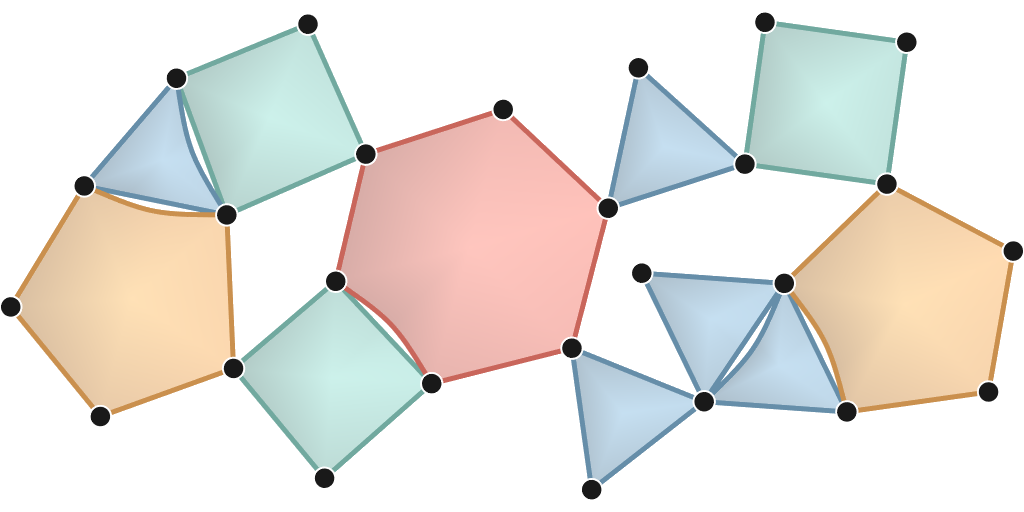}} \\
  \subfloat[][Interior facial cycles $\{F_1,F_2\}$.]{\includegraphics[width=2in]{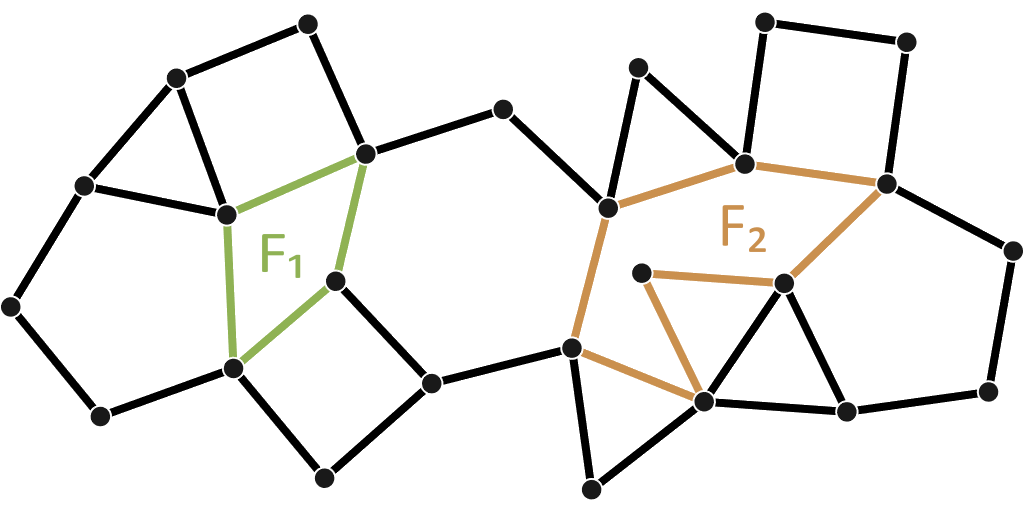}} \\
  \subfloat[][Face triangulation graph $T(H)$.]{\includegraphics[width=2in]{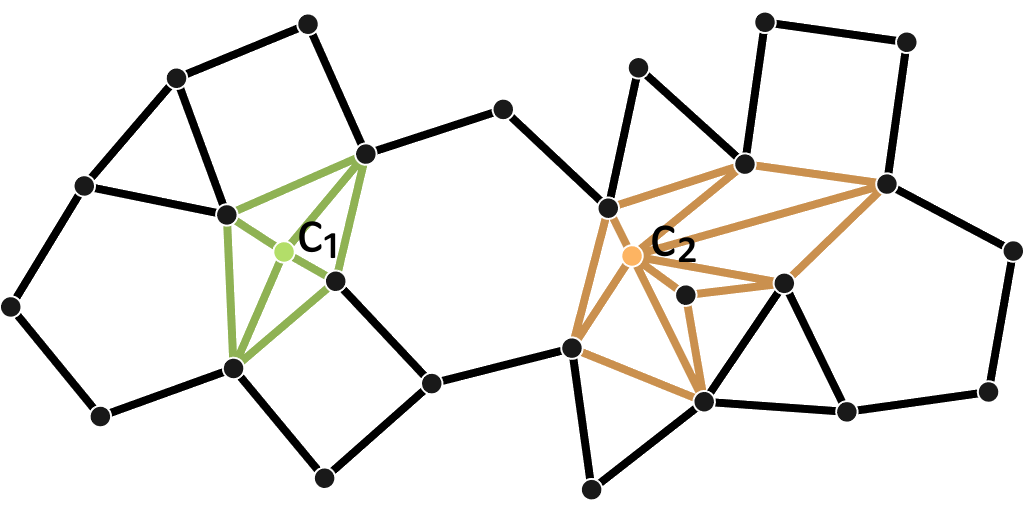}}
  \caption{Construction of the face triangulation graph from a planar polygon drawing of a hypergraph.}\label{fig:triangulation}
\end{figure}

For our polygon layouts of hypergraphs, we require that each hyperedge be drawn as a strictly convex polygon. Thus, we are interested in the case where the face triangulation graph has strictly facial cycles. Thomassen notes that if $\Sigma$ is restricted to being strictly convex, condition (iii) becomes redundant. Thomassen also notes that every vertex $x \in V(G) \setminus V(S)$ with degree 2 must be on a straight line segment in any convex representation of $G$. It follows that the faces whose boundaries include $x$ are not strictly convex. Thus, if we require that every face boundary be strictly convex, Thomassen's Theorem is reduced to the following:

\begin{theorem}[Strictly Convex Representations of Graphs]
  Let $G$ be a biconnected planar graph and let $S$ be a cycle which is the face boundary of some plane representation of $G$. Let $\Sigma$ be a strictly convex polygon representing $S$. Then $\Sigma$ can be extended into a strictly convex representation of $G$ if and only if each vertex $x$ in $G-V(S)$ is joined to $S$ by three paths that are disjoint except for $x$.
  \label{thm:strict}
\end{theorem}

The proof of Theorem~\ref{thm:strict} follows from the proof of Thomassen's Theorem provided in~\cite{thomassen1984plane}. These theorems on convex representations of graphs motivate an extension to convex polygon representations of hypergraphs. We use the face triangulation graph $T(H)$ to connect these theories on graph drawing to hypergraph polygon drawings. First, we must specify the conditions under which a face triangulation graph meets the prerequisites for Theorem~\ref{thm:strict}.

\begin{lemma}\label{lem:planar}
  If the polygon drawing of $H$ corresponds to a plane representation of $K(H)$, then the face triangulation graph $T(H)$ is a planar graph.
\end{lemma}

\begin{proof}
  Since the K{\"o}nig graph $K(H)=(X,Y,D)$ is planar, we know that it does not contain a subgraph homeomorphic to $K_5$ or $K_{33}$. Let us augment $K(H)$ by adding an edge $(u,v)$ for every $u,v \in X$ such that $(u,v)$ forms the side of a polygon drawing of $H$. This augmentation cannot create a subgraph homeomorphic to $K_{33}$ since it does not add any bipartite edges. Notice that each hyperedge in $H$ now corresponds to a wheel subgraph in $K(H)$ which is a planar graph. Since no edges are added between vertices in $Y$, it follows that a subgraph homeomorphic to $K_5$ cannot contain a vertex $y \in Y$ as a non-subdivision vertex. Further, each edge added between a pair of vertices $u,v \in X$ can already be obtained in $K(H)$ by smoothing their common adjacent vertex in $Y$. Thus, our augmentation of $K(H)$ does not affect its planarity. Now let us further augment $K(H)$ by removing all the vertices in $X$ and all the edges in $D$. Then we are left with only the vertices corresponding to hypergraph vertices in $H$ and edges corresponding to the sides of polygons in the drawing of $H$. Clearly, this does not affect the planarity of $K(H)$. We can now obtain the face triangulation graph $T(H)$ by triangulating each interior face of the augmented graph $K(H)$ that corresponds to a hypergraph face in the drawing of $H$. Since an interior face is a bounded region, it follows that we can place the new vertex inside the bounded region and add edges according to step 4 of Procedure~\ref{pro:triangulation} without introducing any edge crossings. Thus, our construction of $T(H)$ is a planar graph.
\end{proof}

\begin{lemma} \label{lem:biconnected}
  If the polygon drawing of $H$ corresponds to a plane representation of $K(H)$, then the face triangulation graph $T(H)$ is biconnected.
\end{lemma}

\begin{proof}
  Each hyperedge in $e \in E(H)$ is replaced by a cycle $C_e$ in $T(H)$ following the sides of the corresponding polygon in the drawing of $H$. Each cycle $C_e$ defines a biconnected subgraph of $T(H)$. Since $H$ is biconnected by assumption, it follows that the union of all cycles $C_e$ for $e \in E(H)$ is also biconnected.
\end{proof}

We now have the appropriate conditions to make the connection between Theorem~\ref{thm:strict} and convex polygon representations of hypergraphs.

\begin{theorem} \label{thm:biconnected}
  Let $H$ be a biconnected, Zykov planar hypergraph. Then $H$ has a convex polygon representation if and only if it has a face triangulation graph $T(H)$ with a strictly convex representation.
\end{theorem}

\begin{figure}[!t]
  \centering
  \subfloat[][]{\includegraphics[height=1.75in]{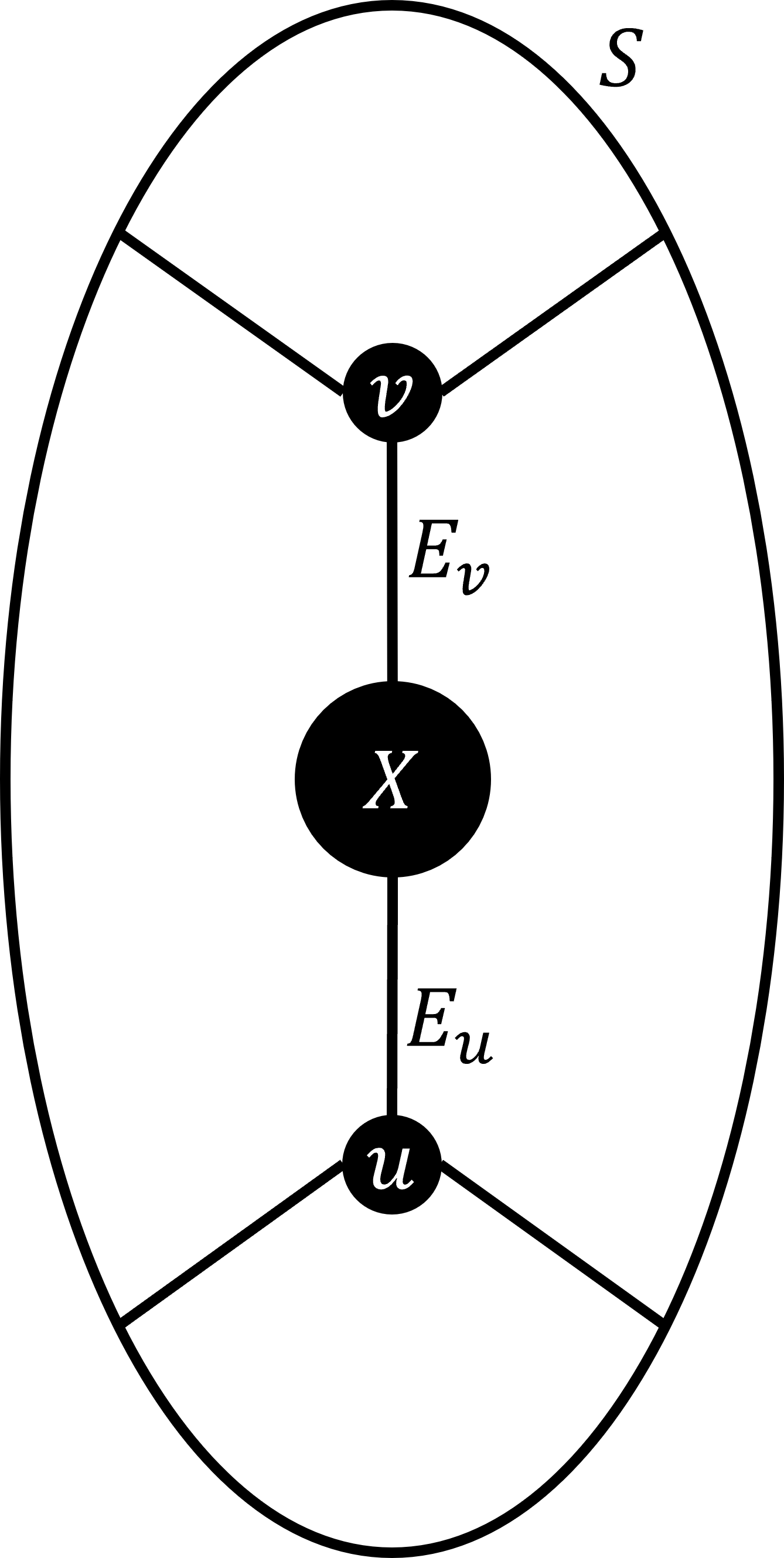}} \hspace{0.125in}
  \subfloat[][]{\includegraphics[height=1.75in]{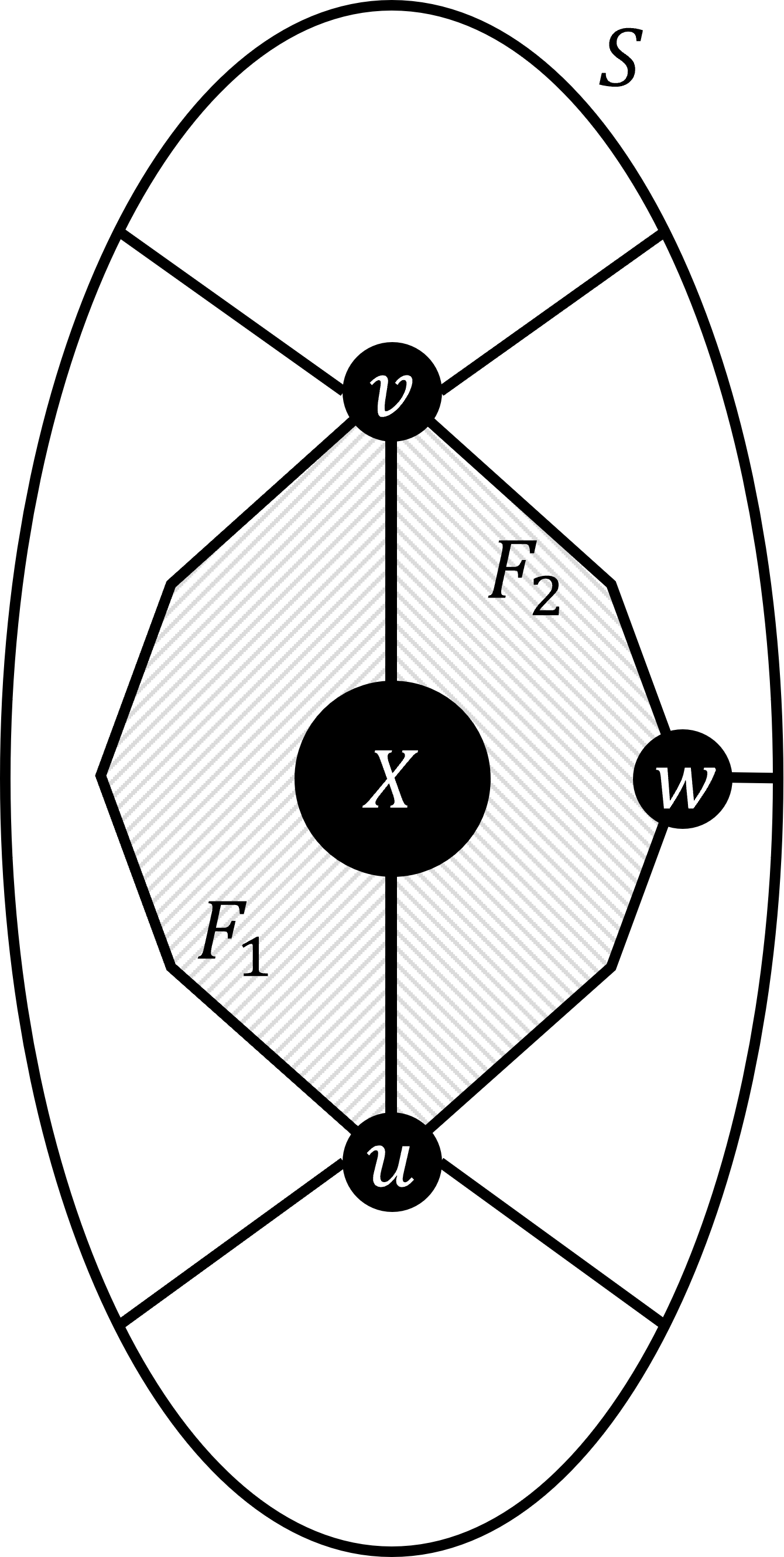}} \hspace{0.125in}
  \subfloat[][]{\includegraphics[height=1.75in]{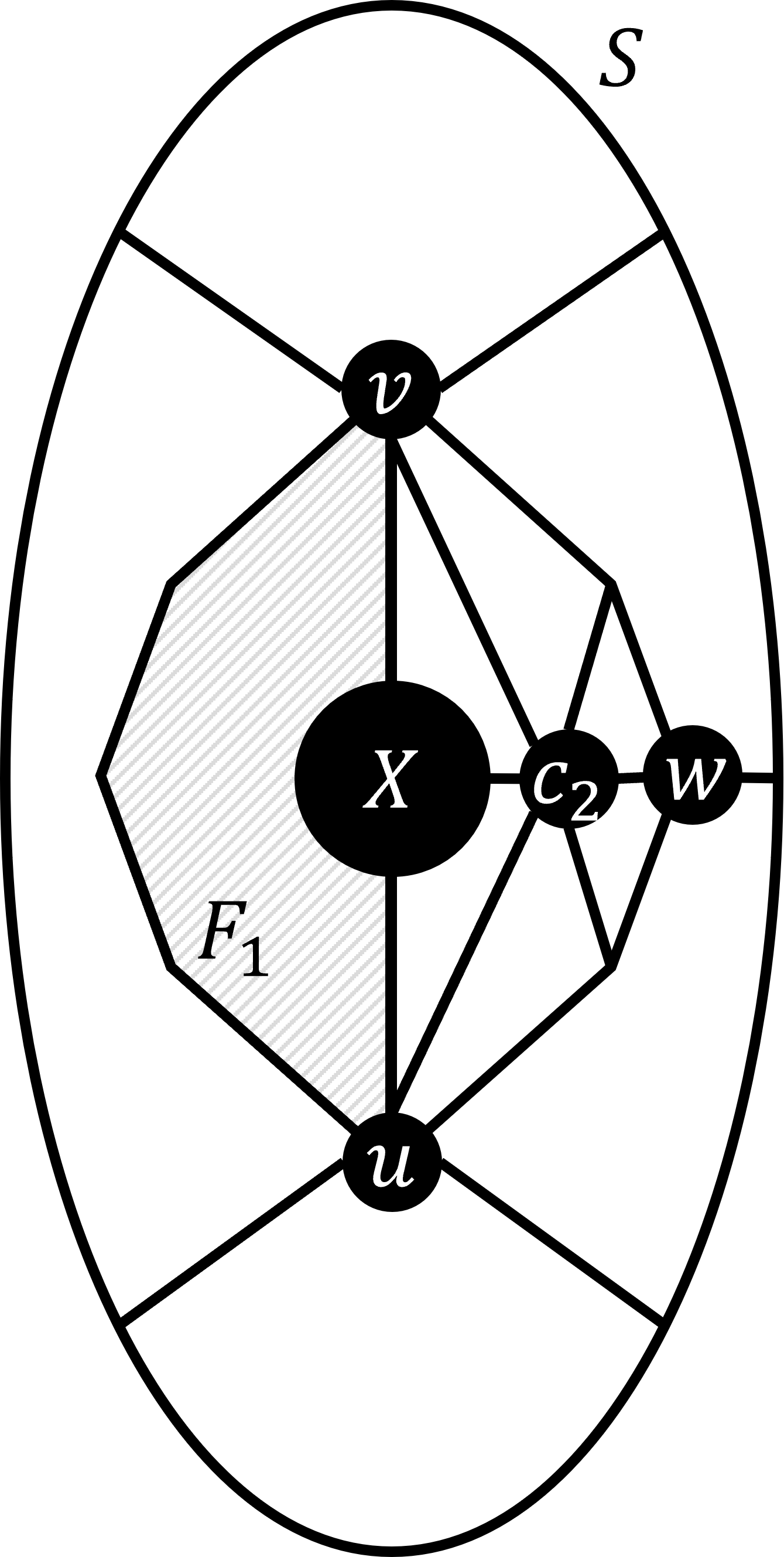}} \\
  \subfloat[][]{\includegraphics[height=1.75in]{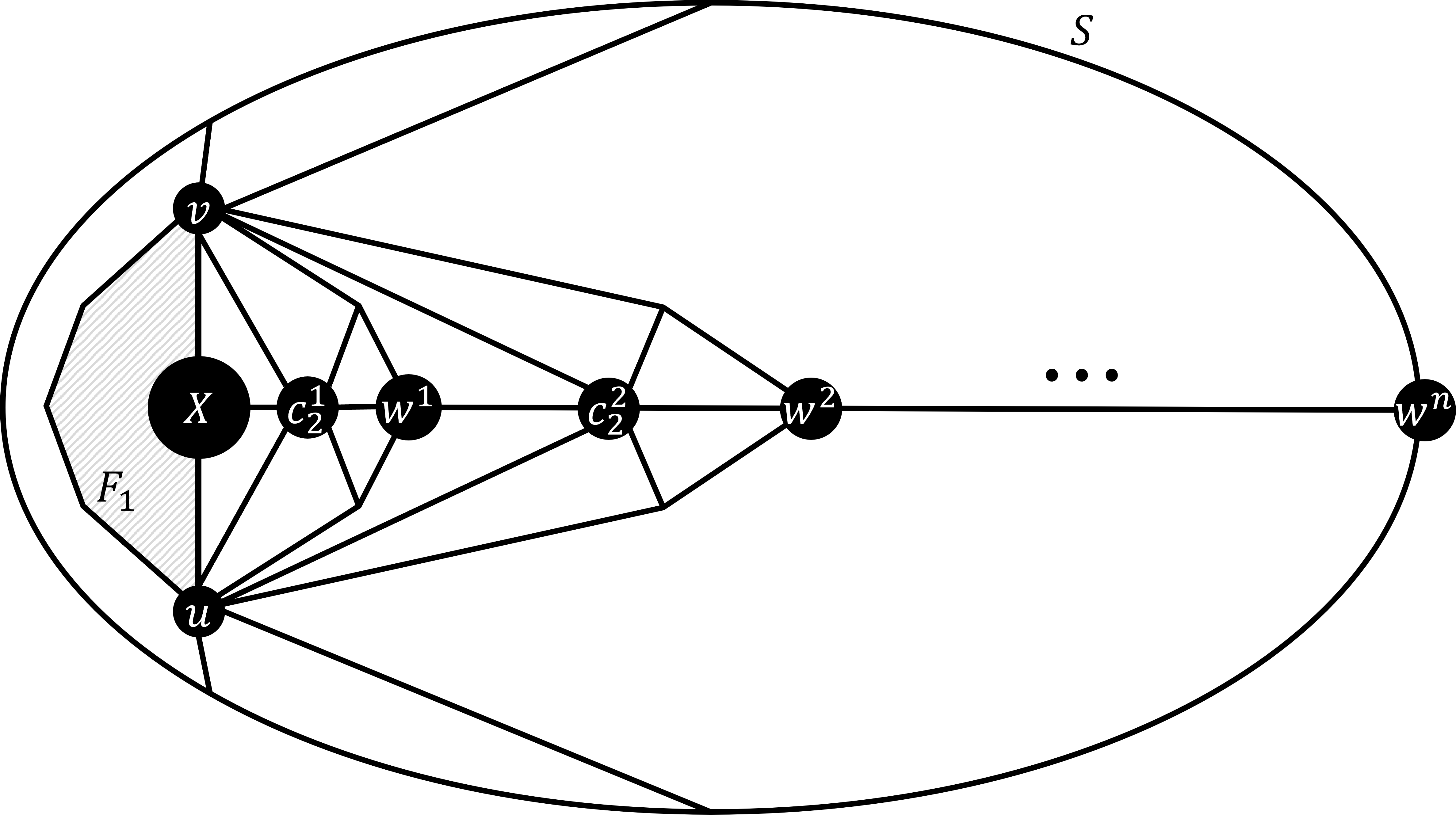}}
  \caption{Illustration of our argument for the proof of Theorem~\ref{thm:biconnected}. In (a), we illustrate the subgraph $X$ in the face triangulation graph which is joined to the exterior face $S$ through two vertices $u$ and $v$. In (b), we illustrate the faces $F_1,F_2$ enclosing $X$. In (c) we illustrate how the triangulation of $F_2$ connects $X$ to a vertex $w$ on $F_2$ through the triangulation vertex $c_2$. In (d) we illustrate how we construct a path from $X$ to $S$ through a sequence of face triangulations $(c_2^1,w^1,c_2^2,w^2,\dots,c_2^n,w^n)$ that does not include $u$ or $v$.}\label{fig:thm3}
\end{figure}

\begin{proof}
  We first show that if $H$ has a face triangulation graph with a strictly convex representation, then it has a convex polygon representation. Suppose there is a drawing of $H$ such that the face triangulation graph $T(H)$ has a strictly convex representation. By the construction of $T(H)$, each hyperedge in $H$ corresponds to a strictly convex facial cycle in the drawing of $T(H)$. It follows that if we draw each hyperedge in $H$ as a polygon following the corresponding facial cycle in $T(H)$, we can obtain a convex polygon representation of $H$.

  Now we show that if $H$ has a convex polygon representation, then it has a face triangulation graph with a strictly convex representation. Suppose that $H$ has a convex polygon representation. Let $T(H)$ be constructed from this representation of $H$ according to Procedure~\ref{pro:triangulation}. By Lemmas~\ref{lem:planar} and \ref{lem:biconnected}, we have that $T(H)$ is planar and biconnected. Let $T(H)$ be drawn according to the convex polygon representation of $H$, and let $S$ be the exterior facial cycle of this drawing. Let $\Sigma$ be a convex polygon representing $S$. Then by Theorem~\ref{thm:strict}, $\Sigma$ can be extended to a graph isomorphic to $T(H)$ with strictly convex facial cycles if and only if each vertex $x \in T(H) - V(S)$ is joined to $S$ by three paths that are disjoint except for $x$. Let $x \in T(H) - V(S)$. Since $T(H)$ is biconnected, there must be at least two vertex-disjoint paths from $x$ to $V(S)$. In the following paragraphs, we prove that $x$ is joined to $S$ through at least three disjoint paths by contradiction.

  Suppose there are two vertices $u,v \in T(H)$ such that every path from $x$ to $V(S)$ includes $u$ or $v$. This implies that if $u$ and $v$ were removed, $T(H)$ would become disconnected, and $x$ would be in a separate connected component from $S$. Let $X$ be the connected component of $T(H)$ containing $x$ when $u$ and $v$ are removed. Let $E_u,E_v$ be the sets of edges between $u$ and $V(X)$, and $v$ and $V(X)$ respectively (\Cref{fig:thm3} (a)). With $T(H)$ drawn according to the convex polygon representation of $H$, it must be that $X$ is drawn in the interior of $S$. It follows that $X$ is enclosed by two faces $F_1$ and $F_2$ in the drawing of $T(H)$ whose boundaries contain $u$ and $v$ (\Cref{fig:thm3} (b)). Notice that $F_1$ and $F_2$ cannot both correspond to strictly convex polygons in the convex polygon representation of $H$ since such polygons would necessarily share a side $(u,v)$. This configuration would preclude $X$ from being incident to both $F_1$ and $F_2$ while also being enclosed by $F_1$ and $F_2$. So, it must be that either $F_1$, $F_2$, or both are faces in the drawing of $T(H)$ that correspond to a hypergraph face in a convex polygon representation of $H$.

  Consider the case where $F_1$ corresponds to a polygon in the convex polygon representation of $H$ and $F_2$ does not. In order for $F_2$ to be drawn as a simple polygon, which must be the case since our drawing of $T(H)$ is planar, it must be that the boundary of $F_2$ contains at least one vertex $w \neq u,v$, $w \notin V(X)$ (\Cref{fig:thm3} (b)). Similarly, in the case where neither $F_1$ nor $F_2$ correspond to a polygon in the convex polygon representation of $H$, it must be that the boundary of either $F_1$ or $F_2$ contains at least one vertex $w \neq u,v$, $w \notin V(X)$. Without loss of generality, suppose that $F_2$ does not correspond to a polygon in the convex representation of $H$, and that the boundary of $F_2$ contains such a vertex $w \neq u,v$, $w \notin V(X)$. In this case, the construction of $T(H)$ would have added a vertex $c_2$ to the interior of $F_2$ and edges $(c_2,y)$ for every vertex $y$ on the boundary of $F_2$. Then $X$ would be connected to the vertex $w$ through a path containing $w$ which contradicts our observation that $w \notin X$ (\Cref{fig:thm3} (c)).

  We can apply the same argument to an updated subgraph $X$ and new faces $F_1$ and $F_2$ enclosing $X$. In this way, we can grow $X$ with a sequence of vertices $(w_1,w_2,\dots,w_n)$ until $w_n$ is a vertex on $S$, which is possible assuming $H$ and $T(H)$ are finite (\Cref{fig:thm3} (d)). Thus, we have shown that there must exist a path from $x$ to $w_n \in V(S)$ that does not contain the vertices $u$ or $v$. This contradicts our assumption that every path from $x$ to $S$ passes through $u$ or $v$. Therefore, each vertex $x \in T(H)-V(S)$ is joined to $S$ by three paths that are disjoint except for $x$, and $\Sigma$ can be extended to a strictly convex representation of $T(H)$ by Theorem~\ref{thm:strict}.
\end{proof}

We also wish to consider convex polygon planarity for hypergraphs that are not biconnected. To do this, we must address the placement of articulation vertices between biconnected hypergraph components, requiring the following lemma.

\begin{lemma} \label{lem:faces}
  Let $H$ be a hypergraph with exterior face $S$ and interior faces $R = \{F_1,F_2,\dots,F_n\}$ for some convex polygon representation of $H$. Then $H$ also has convex polygon representations for each face $F_i \in R$ such that $F_i$ is the exterior face and $R-F_i+S$ are the interior faces.
\end{lemma}

This lemma can be proven in a similar manner to Theorem~\ref{thm:biconnected}. Now we can extend Theorem~\ref{thm:biconnected} to a more general class of connected hypergraphs if we consider each biconnected component individually.

\begin{theorem} \label{thm:components}
  Let $H$ be a Zykov planar hypergraph with $k$ biconnected components. Let $\{B_1,B_2,\dots,B_k\}$ be the sub-hypergraphs induced by the $k$ biconnected components $V(B_i) \subseteq V(H)$. Then $H$ has a convex polygon representation if and only if each sub-hypergraph $B_i$ has a convex polygon representation where every vertex $x \in V(B_i)$ that is also an articulation vertex of $H$ is located on a face boundary of some convex polygon representation of $B_i$.
\end{theorem}

\begin{proof}
  To prove Theorem~\ref{thm:components} in the forward direction, suppose that each sub-hypergraph $B_i$ has a convex polygon representation where every vertex $x \in V(B_i)$ that is also an articulation vertex of $H$ is located on a face boundary of the convex polygon representation of $B_i$. Then we can construct a convex polygon representation for $H$ by starting with the convex polygon representation of $B_i$. Now consider a sub-hypergraph $B_j$ incident to $B_i$ through the articulation vertex $x \in V(H)$. By Lemma~\ref{lem:faces}, $B_j$ has a convex polygon representation where the face boundary containing $x$ is the exterior face boundary. It follows that we can draw $B_j$ with this representation inside the face in $B_i$ whose boundary contains $x$ without introducing any polygon intersections (\Cref{fig:biconnected}). We can repeat this process until each of the biconnected sub-hypergraphs is drawn with an appropriate convex polygon representation.

  To prove Theorem~\ref{thm:components} in the reverse direction, suppose that $H$ has a convex polygon representation. Clearly, a convex polygon representation of any sub-hypergraph $B$ can be obtained by removing the vertices and hyperedges not in $B$ from the convex polygon representation of $H$. Let $x$ be an articulation in $H$ belonging to biconnected sub-hypergraphs $B_i$ and $B_j$. To reach a contradiction, suppose that $x$ is not located on a face boundary of some convex polygon representation of $B_i$. Lemma~\ref{lem:faces} implies that $x$ is not located on a face boundary for any convex polygon representation of $B_i$. It follows that if $B_i$ is drawn with a convex polygon representation, there must be some intersection between a pair of hyperedge polygons $e_i \in E(B_i)$ and $e_j \in E(B_j)$ incident to the articulation vertex $x$. This contradicts our assumption that $H$ has a convex polygon representation.
\end{proof}

\begin{figure}[tbp]
  \centering
  \includegraphics[height=1.5in]{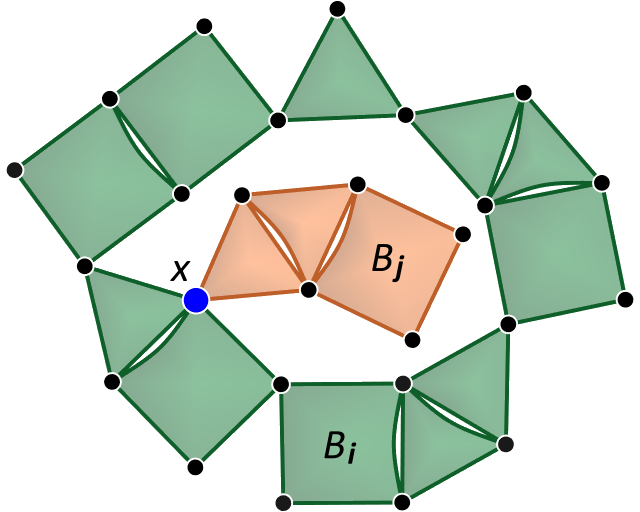}
  \caption{Two biconnected sub-hypergraphs $B_i$ and $B_j$ incident through an articulation vertex $x$ can be drawn without intersection if $x$ is on a face boundary of both sub-hypergraphs.}
  \label{fig:biconnected}
\end{figure}

Finally, we restate our main result for convex polygon representations of hypergraphs.

\begin{theorem}
(Theorem~\ref{thm:forbidden} from the main paper:) Let $H$ be a Zykov planar hypergraph. Then $H$ has a convex polygon representation if and only if it does not contain any of the following as a sub-hypergraph:
  \begin{enumerate}[nosep]
    \item[(a)] A 3-adjacent cluster of 2 hyperedges,
    \item[(b)] A 2-adjacent cluster of 3 hyperedges,
    \item[(c)] A strangled vertex,
    \item[(d)] A strangled hyperedge.
  \end{enumerate}
  \label{thm:forbidden_appendix}
\end{theorem}

\begin{proof}
  $\Longrightarrow$ We prove the contrapositive statement: If $H$ does not have a convex polygon representation, then it contains one of the forbidden sub-hypergraphs. By Theorem~\ref{thm:components}, $H$ does not have a convex polygon representation if the sub-hypergraph induced by one of its biconnected components does not have a convex polygon representation. Let $B$ represent such a biconnected sub-hypergraph. Then by Theorems~\ref{thm:biconnected} and \ref{thm:strict}, it must be that every face triangulation graph $T(B)$ drawn with exterior face boundary $S$ contains a vertex $x \in T(B)-V(S)$ joined to $S$ by fewer than three paths that are disjoint except for $x$. Without loss of generality, let $T(B)$ and $S$ represent the drawing of the face triangulation graph of $B$ containing the fewest such vertices $x$.

  Consider the case where all of the face boundaries containing $x$ represent hyperedges in $B$. If $x$ is on exactly two such face boundaries, it must be that the corresponding hyperedges in $B$ share at least 3 common vertices including $x$. This matches the definition of a 3-adjacent cluster of 2 hyperedges (\Cref{fig:thm4} (a)). If $x$ is on more than two such boundaries, it follows that $x$ is adjacent to at least three other vertices. At least one of these adjacent vertices, call it vertex $y$, must also be joined to $S$ by fewer than three disjoint paths, otherwise, $x$ would be joined to $S$ by three disjoint paths. Without loss of generality, we can consider vertex $y$ instead of vertex $x$, which may be contained in a different set of face boundaries, and could fall under one of the other following cases.

  Now consider the case where $x$ is on a face boundary that neither represents a hyperedge in $H$ nor corresponds to a part of a hypergraph face in $B$. It follows that in the polygon drawing of $B$, $x$ is positioned in the interior of some hyperedge polygon. This indicates the existence of a 2-adjacent hyperedge cluster of 3 hyperedges in $B$ (\Cref{fig:thm4} (b)).

  Now consider the case where $x$ is on a face boundary corresponding to part of a hypergraph face in $B$. Then $x$ is either one of the face triangulation vertices $c$ from step 4 of Procedure~\ref{pro:triangulation}, or is adjacent to such a vertex. If it is the latter, it follows that the adjacent face triangulation vertex $c$ is also joined to $S$ by fewer than three disjoint paths. Without loss of generality, assume that $x=c$. Then $x$ is the central vertex of a wheel subgraph $W$ in $T(B)$. Since $x$ is adjacent to every other vertex in $W$, it follows that $W$ is joined to the rest of $T(B)$ by fewer than three disjoint paths. This configuration corresponds to a strangled hyperedge in $B$ (\Cref{fig:thm4} (c)).

  Thus, we have accounted for all possible configurations of $x$, all of which indicate the existence of a forbidden sub-hypergraph in the biconnected sub-hypergraph $B$. Theorem~\ref{thm:components} further implies that $H$ does not have a convex polygon representation if it contains an articulation vertex $x$ such that $x$ does not appear on a face boundary of any convex polygon representation of some sub-hypergraph $B$ induced by a biconnected component of $H$ containing $x$. Then the hyperedges that are incident to $x$ in $B$  completely surround $x$ in every convex polygon representation of $B$. This can only be possible if the hyperedges incident to $x$ in $B$ form a cycle in $B-x$. This matches the definition of a strangled vertex sub-hypergraph (\Cref{fig:forbidden} (c) from the main paper).

  \begin{figure}[tbp]
    \centering
    \subfloat[][3-adjacent cluster of 2 hyperedges]{\includegraphics[width=1in]{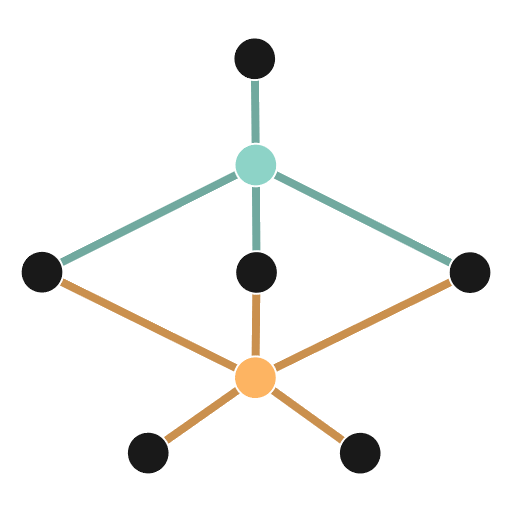} \hspace{0.125in} \includegraphics[width=1in]{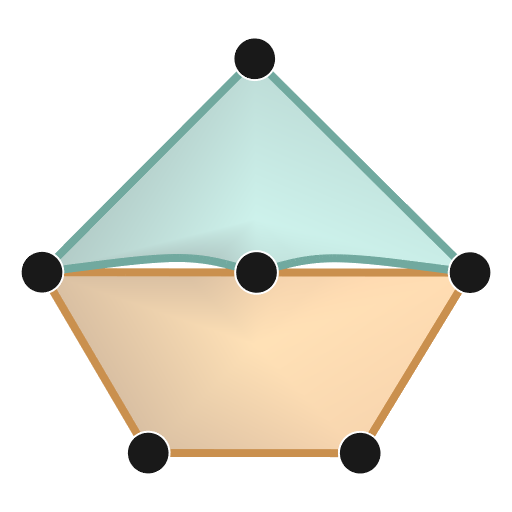} \hspace{0.125in} \includegraphics[width=1in]{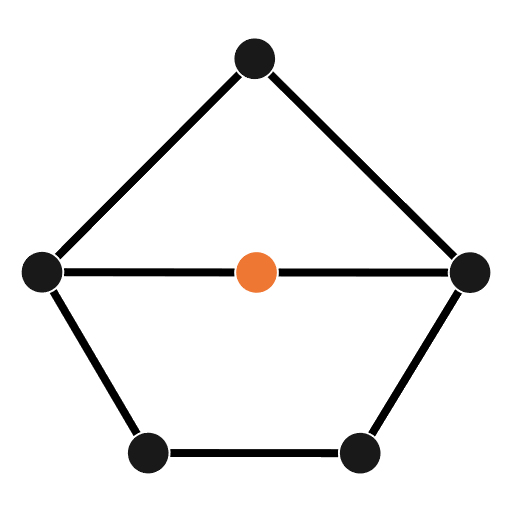}} \\
    \subfloat[][2-adjacent cluster of 3 hyperedges]{\includegraphics[width=1in]{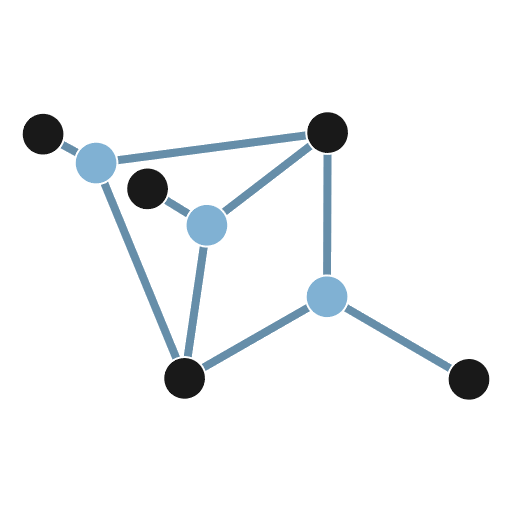} \hspace{0.125in} \includegraphics[width=1in]{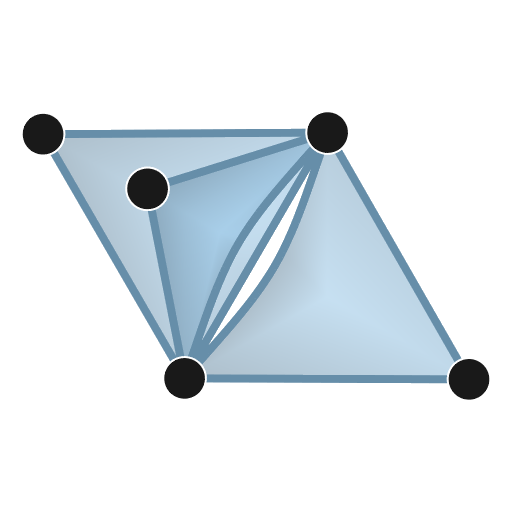} \hspace{0.125in} \includegraphics[width=1in]{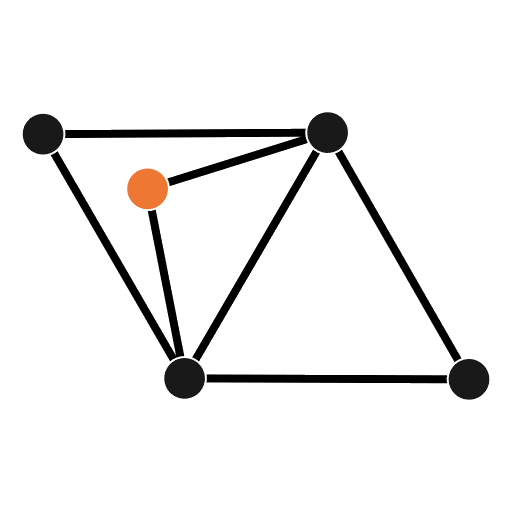}} \\
    \subfloat[][Strangled hyperedge]{\includegraphics[width=1in]{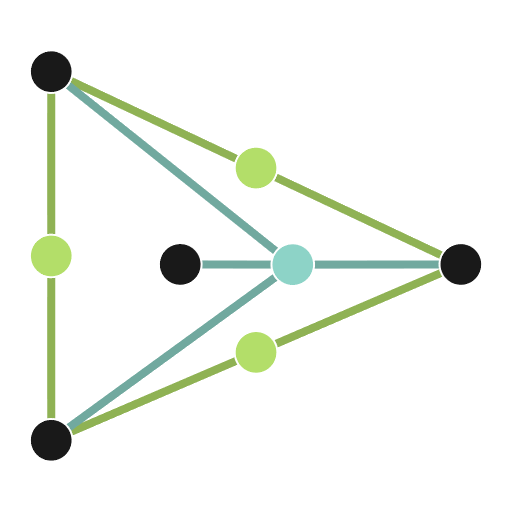} \hspace{0.125in} \includegraphics[width=1in]{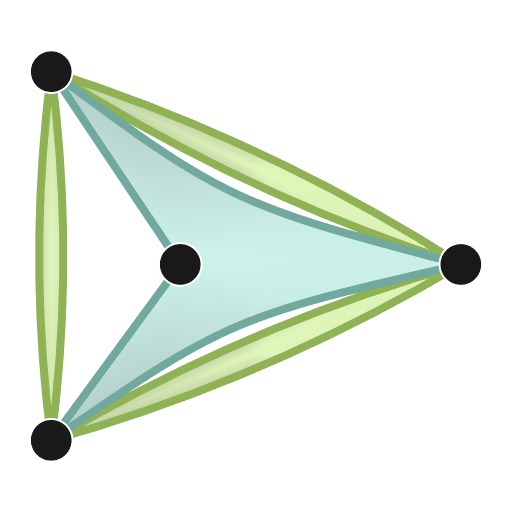} \hspace{0.125in} \includegraphics[width=1in]{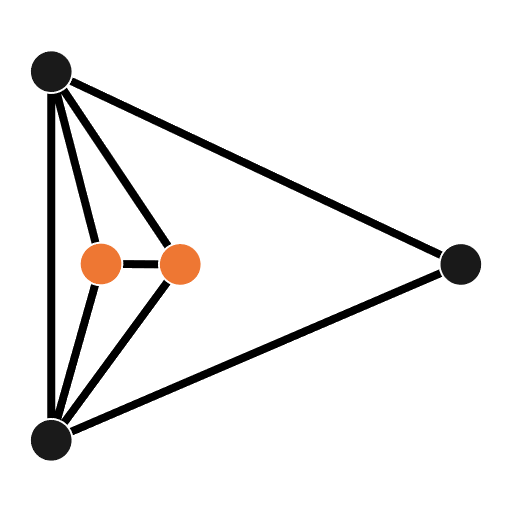}}
    \caption{Forbidden sub-hypergraphs (middle) drawn according to a plane embedding of their K{\"o}nig graphs (left) and their corresponding face triangulation graphs (right). The vertices highlighted in orange have fewer than 3 disjoint paths to the exterior face boundary.}\label{fig:thm4}
  \end{figure}

  $\Longleftarrow$ We prove the contrapositive statement: if $H$ contains any of the forbidden sub-hypergraphs, it does not have a convex polygon representation. First, consider the case where $H$ contains a 3-adjacent cluster of 2 hyperedges. When embedded in the plane, the three shared vertices in the 3-adjacent cluster must either form a triangle or be colinear. If they form a triangle, the intersection of two convex polygons containing the vertices must at least equal the area of the triangle. If the vertices are colinear, then the polygons containing them are not strictly convex. Thus, $H$ does not have a convex polygon representation.

  Now consider the case where $H$ contains a 2-adjacent cluster of 3 hyperedges. Then the locations of the two shared vertices in the cluster define a line splitting the plane into two half-planes. Let $\{e_1,e_2,e_3\}$ be the three hyperedges in the cluster. Without loss of generality, $e_1$ can be drawn with its remaining vertices in one half plane, and $e_2$ can be drawn with its remaining vertices in the other half plane, and there is no intersection between the polygons for $e_1$ and $e_2$. For $e_3$ to be drawn as a convex polygon, its remaining vertices must be drawn in one half plane or the other, so it must have a nonzero intersection with the polygon for $e_1$ or $e_2$, and $H$ does not have a convex polygon representation.

  Now consider the case where $H$ contains a strangled vertex $x \in V(H)$. Let $C$ be the cycle among a proper subset of the vertices adjacent and hyperedges incident to $x$. If the vertices $V(C)$ are positioned in the plane such that their convex hull does not match their order in the cycle, it must be that the cycle crosses over itself and the drawing is non-planar. Otherwise, if $x$ is located outside the convex hull of $V(C)$, it must be that one or more of the hyperedges in $E(C)$ have polygons crossing the interior and boundary of the convex hull, so the drawing is non-planar. If $x$ is located inside the convex hull of $V(C)$, it follows that if each hyperedge in $E(C)$ is drawn as a convex polygon, the interior of the hull is completely tiled by these polygons. Thus, any other hyperedge polygon incident to $x$ must intersect with one of the hyperedge polygons in $E(C)$, so $H$ does not have a convex polygon representation.

  Now consider the case where $H$ contains a strangled hyperedge $e \in E(H)$. Let $C$ be the cycle among a proper subset of the vertices incident and hyperedges adjacent to $e$. If the vertices $V(C)$ are located such that their convex hull does not match their order in the cycle, it must be that the cycle crosses over itself and the drawing is non-planar. Otherwise, if a vertex $x$ incident to $e$ but not in $V(C)$ is drawn inside the convex hull of $V(C)$, it follows that $e$ cannot be drawn as a strictly convex polygon. If $x$ is drawn outside the convex hull of $V(C)$, it follows that the drawing of $e$ has nonzero intersection with at least one hyperedge polygon in $E(C)$, so $H$ does not have a convex polygon representation.
\end{proof}

\section{Paper-Author Results}
\label{sec:paperauthor}

\Cref{fig:paperauthor_enlarged}: An enlarged version of \Cref{fig:paperauthor} from the main paper. A paper-author hypergraph dataset containing 1008 vertices and 429 hyperedges (a) is simplified with our framework down to the coarsest allowable scale $H_{1214}$ (c) and the layout is optimized. Then the simplification is iteratively reversed, and the layout refined at each intermediate scale, an example of which is shown in (b), until the original scale $H_0$ is reached.

\section{Eye Tracking Gaze Paths}
\label{sec:gazetrajectory}

\Cref{fig:gazetrajectory_enlarged}: An enlarged version of \Cref{fig:gazetrajectory} from the main paper. Gaze fixation paths of two participants answering the same question for the trade agreement dataset in our user survey. The participant with the gaze path on the left did not study the dual hypergraph and answered the question incorrectly. The participant with the gaze path on the right studied both the primal and dual hypergraph visualizations and answered the question correctly.

\section{Eye Tracking Fixation Timelines}
\label{sec:gazetimeline}

\Cref{fig:gazetimeline_enlarged}: An enlarged version of \Cref{fig:gazetimeline} from the main paper. Gaze fixation timelines of two participants answering a question in our user survey. The vertical axis indicates different regions on the participant's screen, including the question text and visualization scales. The horizontal axis represents the time in seconds that a participant spent on the question. The vertical lines in the plot indicate when the participant selected an answer. The blue lines indicate a correct answer and the red lines indicate an incorrect answer.

\begin{figure*}[tbp]
  \centering
  \subfloat[][Original Scale]{\includegraphics[width=5in]{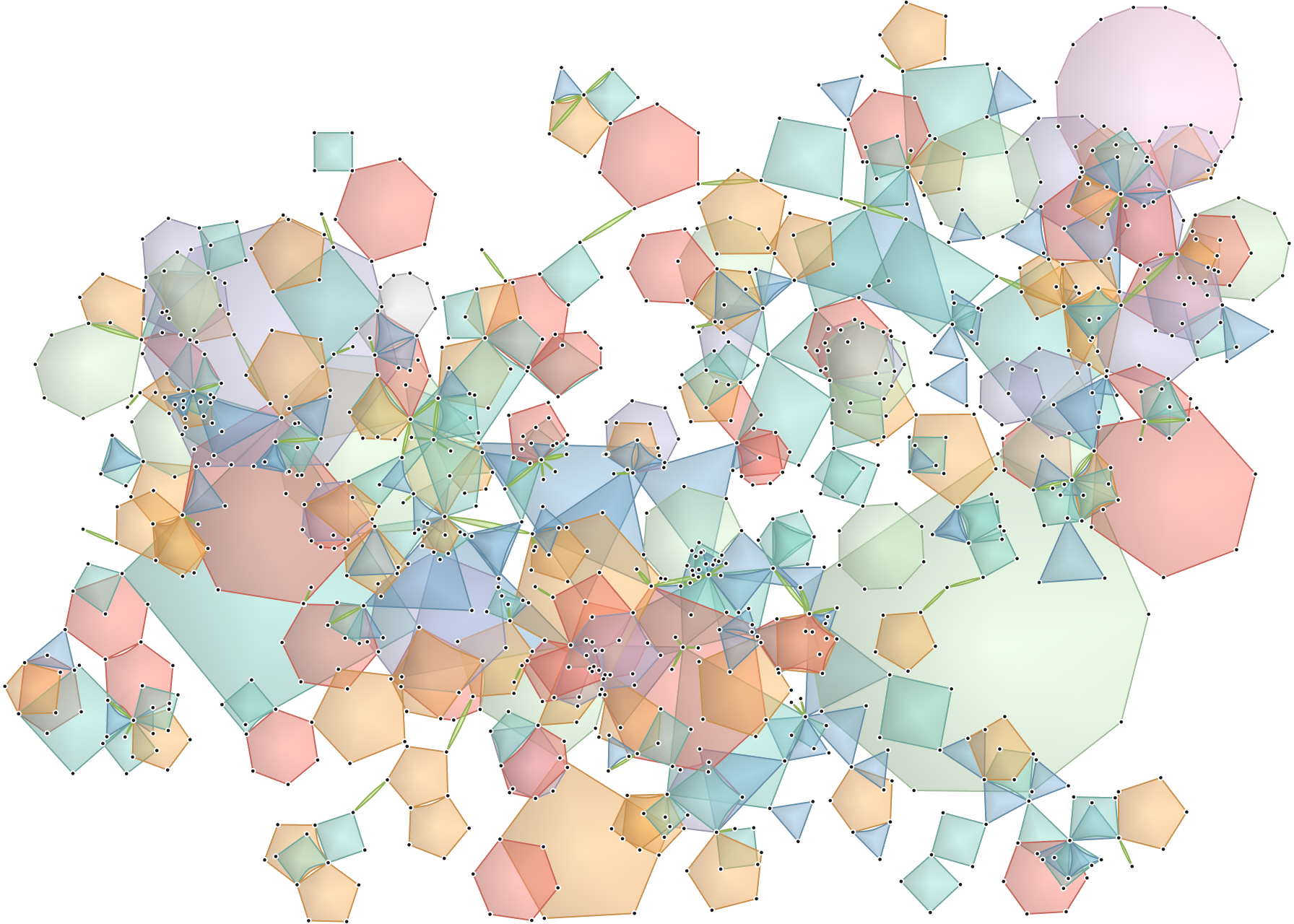}} \\
  \vspace{0.22in}
  \subfloat[][Intermediate Simplified Scale]{\includegraphics[width=5in]{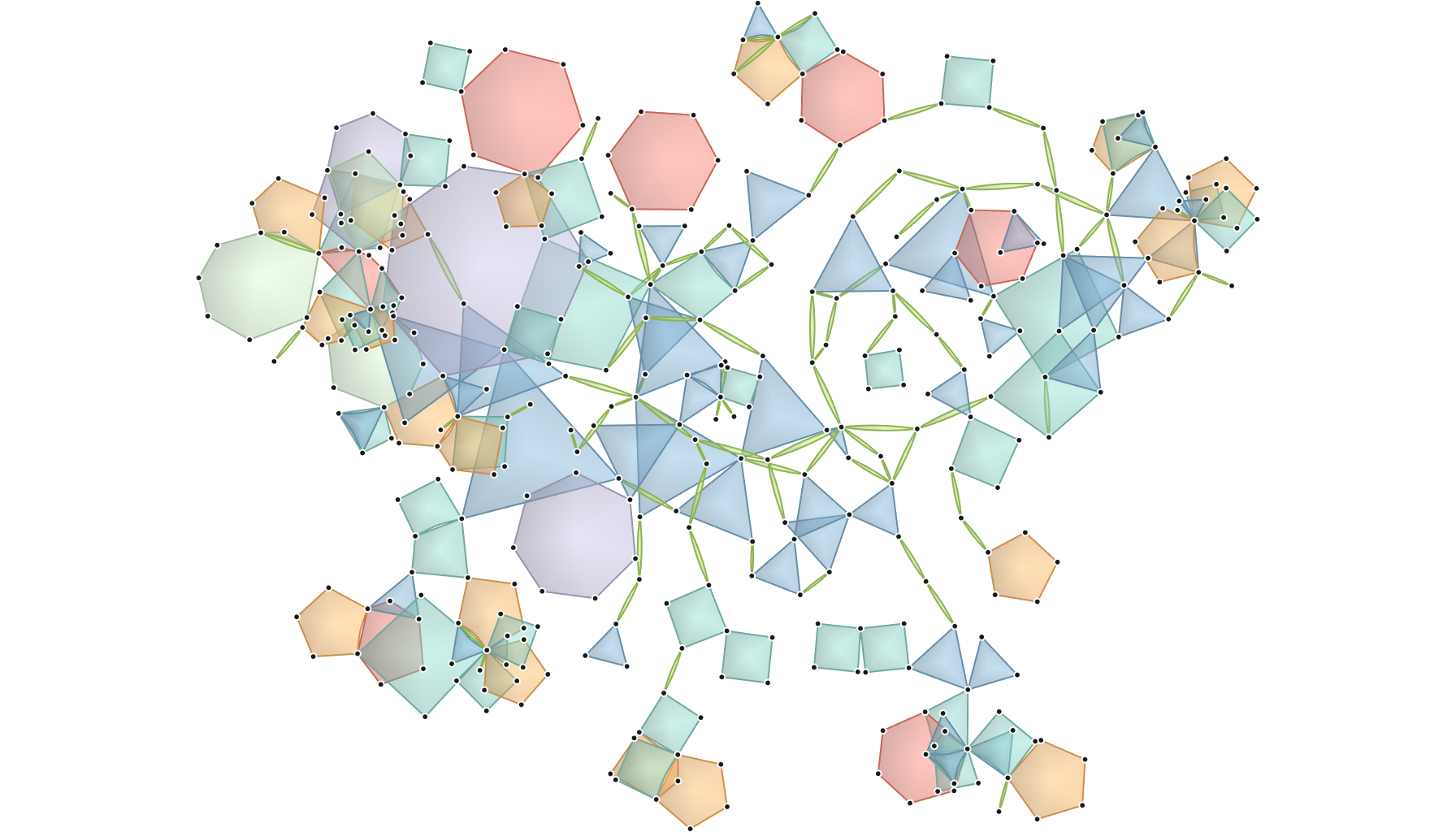}} \\
  \vspace{0.22in}
  \subfloat[][Coarsest Simplified Scale]{\includegraphics[width=5in]{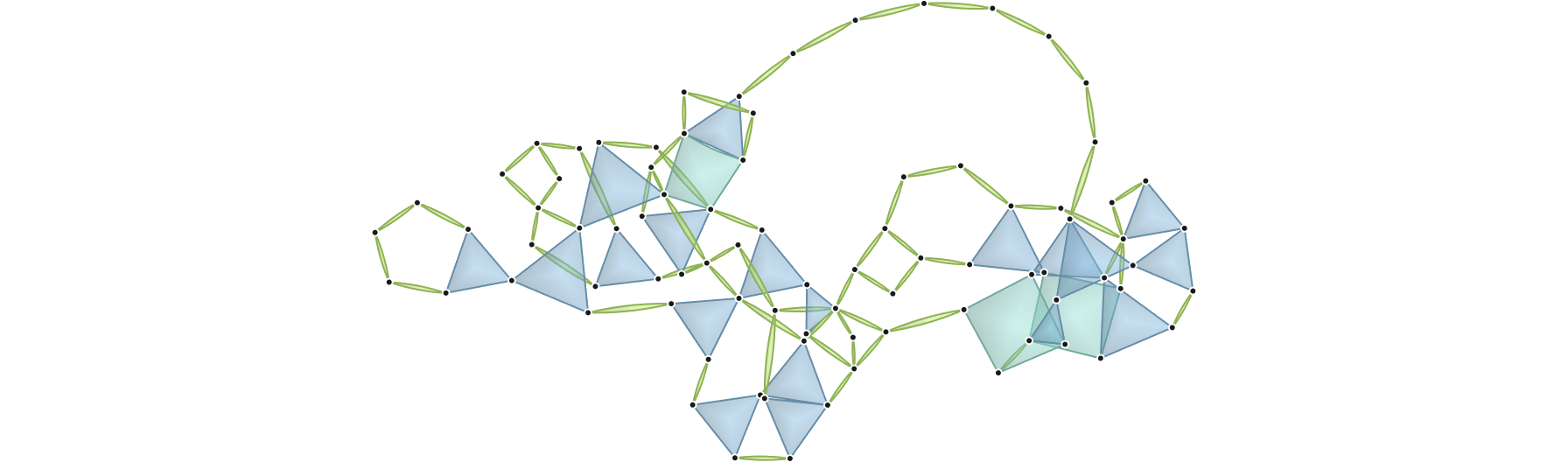}}
  \caption{Enlarged versions of the images in \Cref{fig:paperauthor} from the main paper. Final optimized layout of the original scale (a), coarsest simplified scale (c), and one intermediate simplified scale (b) for a paper-author hypergraph dataset.}\label{fig:paperauthor_enlarged}
\end{figure*}

\begin{figure*}[tbp]
  \centering
  \includegraphics[width=3.5in]{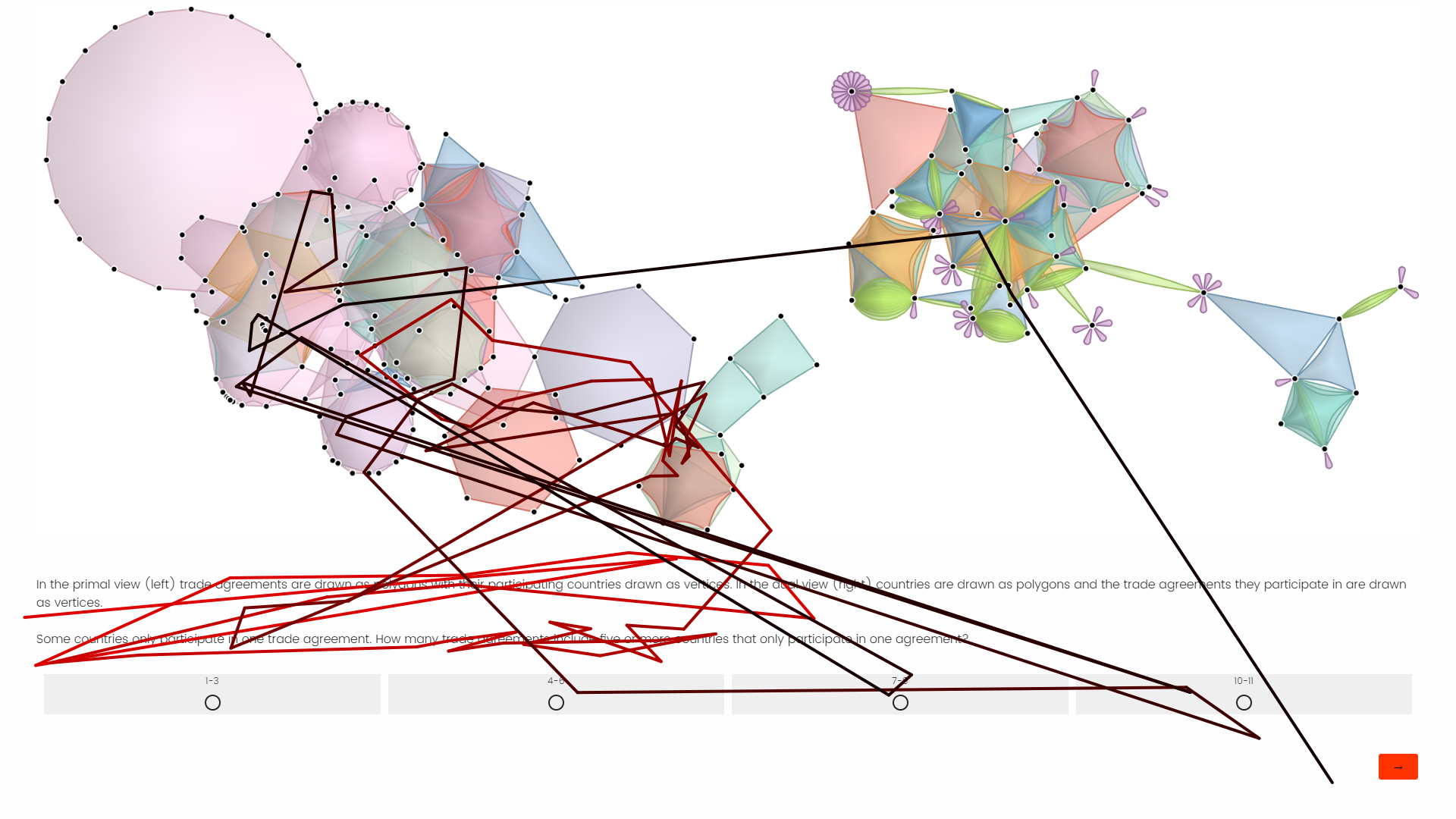}
  \includegraphics[width=3.5in]{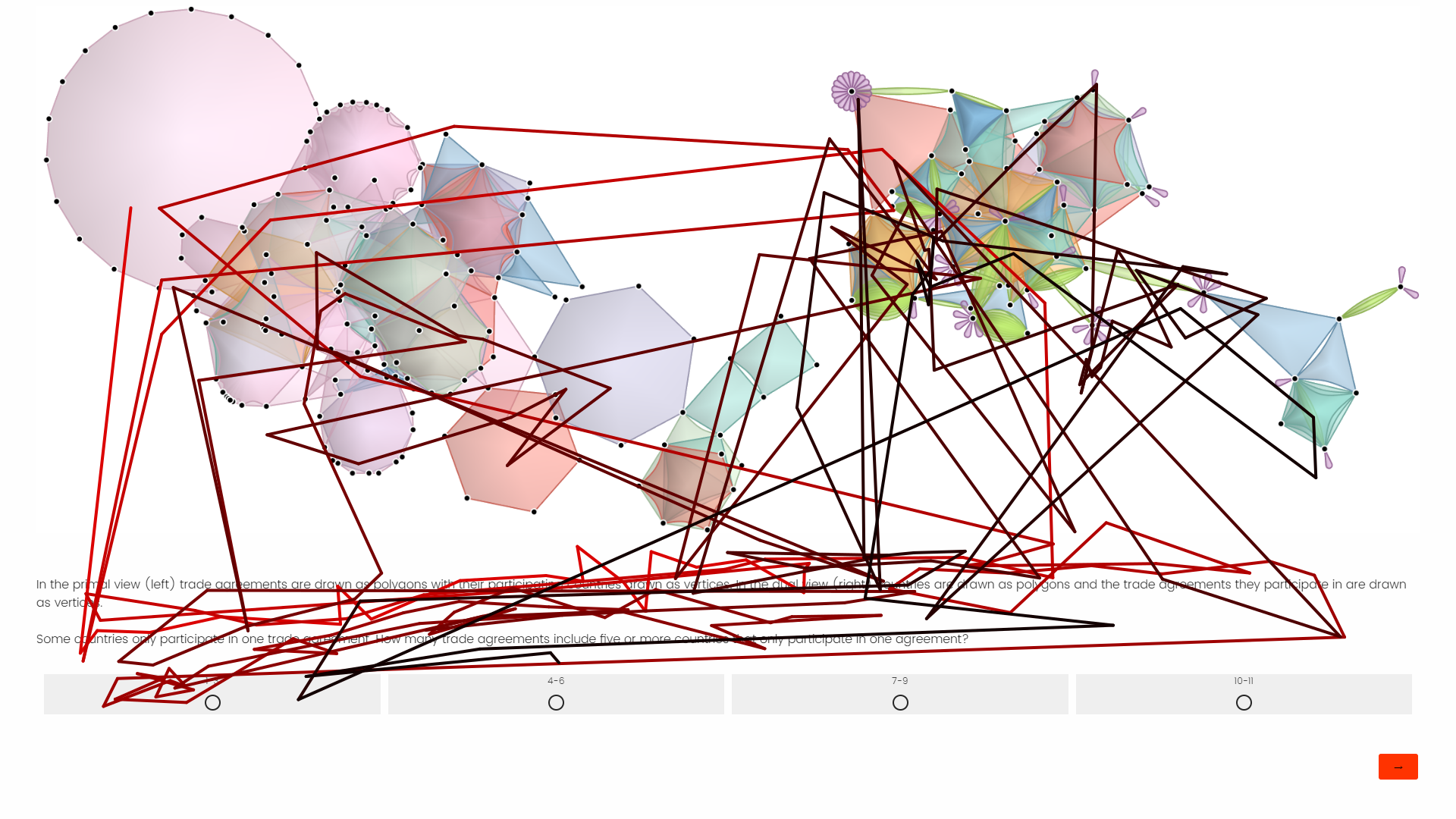}
  \caption{Enlarged versions of the images in \Cref{fig:gazetrajectory} from the main paper. Gaze fixation paths of two user survey participants.}\label{fig:gazetrajectory_enlarged}
\end{figure*}

\begin{figure*}[tbp]
  \centering
  \includegraphics[width=7in]{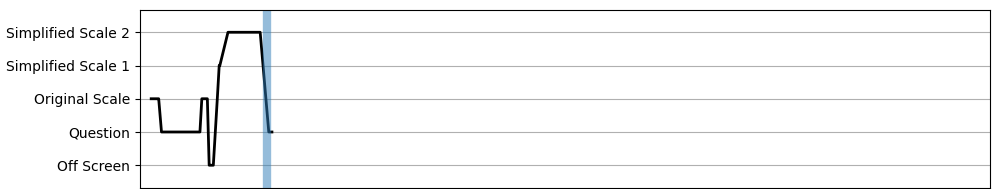} \\
  \includegraphics[width=7in]{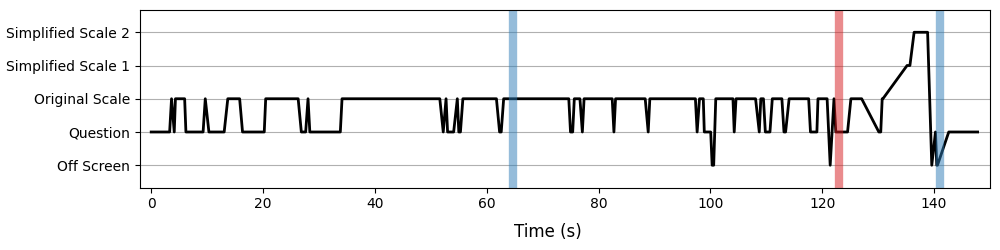}
  \caption{Enlarged versions of the images in \Cref{fig:gazetimeline} from the main paper. Gaze fixation timelines of two user survey participants.}
  \label{fig:gazetimeline_enlarged}
\end{figure*}

\end{document}